\documentclass[a4paper,american]{lipics-v2018-modified}

\usepackage{microtype}

\usepackage[utf8]{inputenc}
\usepackage{amsmath}
\usepackage{amssymb}
\usepackage{amsthm}
\usepackage{amsfonts}
\usepackage{url}
\usepackage{xypic}
\usepackage{xspace}
\usepackage{graphicx}
\usepackage{tabu}
\usepackage{xcolor}
\usepackage{wrapfig}

\usepackage{xstring}

\def\nth#1{
    \IfInteger{#1}
    {\nthH{#1}}
    {#1\nthscript{th}}
}

\def\nthH#1{
    \expandafter\nthM \number #1\relax
    \nthscript{%
        \ifnum#1\nthtest0 th\else 
        \expandafter \nthSuff \expandafter 0\number\ifnum #1<0-\fi#1\delimiter
        \fi
    }}
    
\def\nthM#1{\if -#1\ifmmode-\else$-$\fi\else #1\fi}

\def\nthSuff#1#2#3{%
    \ifx \delimiter#3
    \ifnum #1=1 th
    \else 
    \ifcase #2 th\or st\or nd\or rd\else th\fi
    \fi
    \else 
    \expandafter \nthSuff \expandafter #2\expandafter #3%
    \fi}

\def\nthscript#1{\textsuperscript{#1}}

\def\nthtest{=}  

\theoremstyle{plain}
\newtheorem{proposition}[theorem]{Proposition}

\usepackage{thmtools}
\usepackage{thm-restate}

\newenvironment{claiminproof}
{%

\vspace{2mm}
\noindent
\textit{Claim:}
}%
{%
\mbox{}%

\vspace{2mm}%
\noindent%
\ignorespacesafterend
}%

\newcommand{\rot}[1]{\rotatebox[origin=c]{90}{$#1$}}

\newcommand{\eqdef}{\stackrel{\text{def}}{=}}
\newcommand{\wqo}{\text{wqo}\xspace}

\newcommand{\W}{\mathcal{W}}
\newcommand{\barW}{\overline{\mathcal{W}}}

\newcommand{\barS}{\overline{S}}
\newcommand{\barT}{\overline{T}}
\newcommand{\barI}{\overline{I}}
\newcommand{\barF}{\overline{F}}
\newcommand{\barX}{\overline{X}}

\newcommand{\barpreceq}{\,\overline{\preceq}\,}

\newcommand{\A}{\mathcal{A}}
\newcommand{\B}{\mathcal{B}}

\newcommand{\N}{\mathbb{N}}

\newcommand{\idec}[2]{\text{\sc Id-dec}_{#1}(#2)}
\newcommand{\idcompl}[1]{\overline{#1}}  

\newcommand{\prodW}{\W_{\times}}
\newcommand{\prodS}{S_{\times}}
\newcommand{\prodT}{T_{\times}}
\newcommand{\prodI}{I_{\times}}
\newcommand{\prodF}{F_{\times}}
\newcommand{\prodpreceq}{\preceq_{\times}}

\newcommand{\WSTS}{\text{WSTS}\xspace}
\newcommand{\UWSTS}{\text{UWSTS}\xspace}
\newcommand{\DWSTS}{\text{DWSTS}\xspace}
\newcommand{\WSTSes}{WSTS\xspace}
\newcommand{\UWSTSes}{UWSTS\xspace}
\newcommand{\DWSTSes}{DWSTS\xspace}
\newcommand{\ULTS}{ULTS\xspace}
\newcommand{\ULTSes}{ULTS\xspace}

\newcommand{\fin}{\textup{fin}}

\newcommand{\pfindown}[1]{{\mathcal{P}}^{\downarrow}_{\fin}(#1)}
\newcommand{\pdown}[1]{{\mathcal{P}}^{\downarrow}(#1)}

\newcommand{\pup}[1]{{\mathcal{P}}^{\uparrow}(#1)}

\newcommand{\upclosure}[1]{{\uparrow}{#1}}
\newcommand{\downclosure}[1]{{\downarrow}{#1}}

\newcommand{\downclosureinv}[1]{{\downarrow_{-1}}{#1}}
\newcommand{\rev}{\textup{rev}}

\newcommand{\trans}[1]{\xrightarrow{#1}}

\newcommand{\revsubseteq}{\subseteq_\text{\tiny rev}}

\newcommand{\succe}[3]{\text{\sc Succ}_{#1}(#2, #3)}
\newcommand{\prede}[3]{\text{\sc Pred}_{#1}(#2, #3)}
\newcommand{\reachgen}[3]{\text{\sc Reach}_{#1}(#2, #3)}
\newcommand{\revreachgen}[3]{\text{\sc Reach}^{-1}_{#1}(#2, #3)}
\newcommand{\revreachall}[1]{\text{\sc Reach}^{-1}_{#1}}
\newcommand{\reach}[2]{\text{\sc Reach}_{#1}(#2)}
\newcommand{\revreach}[2]{\text{\sc Reach}^{-1}_{#1}(#2)}
\newcommand{\reachall}[1]{\text{\sc Reach}_{#1}}
\newcommand{\existsformula}[2]{\exists #1 \colon #2}
\newcommand{\set}[1]{ \left\{ #1 \right\}}
\newcommand{\setof}[2]{\{ #1 \mid #2\}}

\newcommand{\fire}[1]{[#1\rangle{}}
\newcommand{\iso}{\overset{\mathit{iso}}{=}}
\newcommand{\inc}{\mathit{inc}}
\newcommand{\dec}{\mathit{dec}}
\newcommand{\ninc}{N_{\inc}}
\newcommand{\ndec}{N_{\dec}}
\newcommand{\complementof}[1]{#1^{\mathcal{C}}}
\newcommand{\complA}{\complementof{\A}}
\renewcommand{\L}{\mathcal{L}}
\newcommand{\K}{\mathcal{K}}
\newcommand{\R}{\mathcal{R}}

\newcommand{\card}[1]{\left|#1\right|}
\newcommand{\norm}[1]{\lVert#1\rVert_\infty}

\newcommand {\dotcup}{\ensuremath{\mathop{\mathaccent\cdot\cup}}}

\newcommand{\lang}[1]{\mathcal{L}(#1)}

\newcommand{\Nat}{\mathbb{N}}
\newcommand{\EXPSPACE}{\mathsf{EXPSPACE}}
\newcommand{\PSPACE}{\mathsf{PSPACE}}
\newcommand {\bigO} [1]{ \mathcal{O} \left( {#1} \right) }

\newcommand{\NBfree}{N_{\mathit{det}}}
\newcommand{\NAlambda}{N_{-\lambda}}

\newcommand{\interval}[1]{[1..#1]}

\newcommand{\NA}{N_1}
\newcommand{\NB}{N_2}
\newcommand{\TA}{T_1}
\newcommand{\TB}{T_2}

\newcommand{\lambdaA}{\lambda_1}

\newcommand{\wstsof}[1]{\W_{#1}}
\newcommand{\WSTSN}{\wstsof{N}}

\newcommand{\Cmplment}{X}

\newcommand{\binary}{\set{0,1}}
\newcommand{\pout}{p_\mathit{out}}
\newcommand{\pin}{p_\mathit{in}}
\newcommand{\phaltinc}{p_\mathit{haltinc}}
\newcommand{\phaltdec}{p_\mathit{haltdec}}

\newcommand{\llipton}{b}
\newcommand{\lphase}{c}
\newcommand{\calB}{\mathcal{B}}

\title{Regular Separability of Well-Structured Transition Systems}

\author{Wojciech Czerwi\'nski}
{University of Warsaw, Poland}
{wczerwin@mimuw.edu.pl}
{https://orcid.org/0000-0002-6169-868X}
{Supported by the Polish National Science Centre under grant 2016/21/D/ST6/01376.}

\author{S{\l}awomir Lasota}
{University of Warsaw, Poland}
{sl@mimuw.edu.pl}
{https://orcid.org/0000-0001-8674-4470}
{Partially supported by the European Research Council (ERC) project Lipa under the EU Horizon 2020 research and innovation programme (grant agreement No.~683080).}

\author{Roland Meyer}
{TU Braunschweig, Germany}
{roland.meyer@tu-bs.de}
{https://orcid.org/0000-0001-8495-671X}
{}

\author{Sebastian Muskalla}
{TU Braunschweig, Germany}
{s.muskalla@tu-bs.de}
{https://orcid.org/0000-0001-9195-7323}
{}

\author{K Narayan Kumar}
{Chennai Mathematical Institute and UMI RELAX, India}
{kumar@cmi.ac.in}
{}
{Partially supported by the Indo-French project AVeCSo, the Infosys Foundation, and DST-VR Project P-02/2014.}

\author{Prakash Saivasan}
{TU Braunschweig, Germany}
{p.saivasan@tu-bs.de}
{}
{}

\authorrunning{W. Czerwi\'nski, S. Lasota, R. Meyer, S. Muskalla, K. Narayan Kumar, and P. Saivasan}

\Copyright{Wojciech Czerwi\'nski, S{\l}awomir Lasota, Roland Meyer, Sebastian Muskalla, K Narayan Kumar, and Prakash Saivasan}

\subjclass{\\
    \\
    \\
    \\
}


\acknowledgements{\\
    We thank an anonymous referee for pointing out Part~(2) of Corollary~\ref{Corollary:Separability}.\\
    We thank Sylvain Schmitz for helpful discussions.
}

\begin{document}

\maketitle


\begin{abstract}
    We investigate the languages recognized by well-structured transition systems (\WSTS) with upward and downward compatibility.
    Our first result shows that, under very mild assumptions, every two disjoint \WSTS languages are regular separable:
    There is a regular language containing one of them and being disjoint from the other.
    As a consequence, if a language as well as its complement are both recognized by \WSTS, then they are necessarily regular.
    In particular, no subclass of WSTS languages beyond the regular languages is closed under complement.
    Our second result shows that for Petri nets, the complexity of the backwards coverability algorithm yields a bound on the size of the regular separator.
    We complement it by a lower bound construction.
\end{abstract}

\section{Introduction}
We study the languages recognized by well-structured transition systems (\WSTS)~\cite{DBLP:conf/icalp/Finkel87,DBLP:journals/iandc/Finkel90,AJ1993,ACJT96,DBLP:journals/tcs/FinkelS01}. 
\WSTS form a framework subsuming several widely-studied models, like Petri nets~\cite{Esparza:1996:DCP:647444.727062} and their extensions with transfer~\cite{DFS98}, data~\cite{RV12}, and time~\cite{ADM04}, graph rewriting systems~\cite{JK2008}, depth-bounded systems~\cite{M2008,WZH10,EDOsualdo}, ad-hoc networks~\cite{ADRST11}, process algebras~\cite{BGZ04}, 
lossy channel systems (LCS)~\cite{AJ1993}, and programs running under weak memory models~\cite{ABBM10,ABBM12}. 
Besides their applicability, the importance of \WSTS stems from numerous decidability results.
Finkel showed the decidability of termination and boundedness~\cite{DBLP:conf/icalp/Finkel87,DBLP:journals/iandc/Finkel90}.
Abdulla came up with a backward algorithm for coverability~\cite{AJ1993}, for which a matching forward procedure was found only much later~\cite{GEERAERTS2006180}. 
Several simulation and equivalence problems are also decidable for \WSTS~\cite{DBLP:journals/tcs/FinkelS01}. 
The work on \WSTS\ even influenced algorithms for regular languages~\cite{Antichains2006} and recently led to the study of new complexity classes~\cite{SchmitzSchnoebelen2011}. 

Technically, a \WSTS\ is a transition system equipped with a quasi order on the configurations that satisfies two properties.
It is a well quasi order and it is (upward or downward) compatible with the transition relation in the sense that it forms a simulation relation. 
For our language-theoretic study, we assume the transitions to be labeled and the \WSTS to be equipped with sets of initial and final configurations. 
The set of final configurations is supposed to be upward or downward closed wrt. the quasi order of the \WSTS.
When specialized to VAS, this yields the so-called \emph{covering languages}. 

For \WSTS languages, we study the problem of regular separability.
Given two languages $\L$ and $\K$ over the same alphabet, a
\emph{separator} is a language $\R$ that contains one of the languages and is disjoint from the other, $\L \subseteq \R$ and $\R\cap \K = \emptyset$. 
The separator is regular if it is a regular language. 
Separability has recently attracted considerable attention.
We discuss the related work in a moment.

Disjointness is clearly necessary for regular separability. 
We show that for most {\WSTS}, disjointness is also sufficient. 
Our main result is the following:
\begin{center}
Any two disjoint \WSTS languages are regular separable. 
\end{center}
The only assumption we need is that, in the case of upward-compatible \WSTS resp.~downward-compatible \WSTS, one of the \WSTS is finitely branching resp.~deterministic.

The proof proceeds in two steps.
In the first step, we link inductive invariants from verification~\cite{MP1995} to separability in formal languages.
More precisely, we show that any inductive invariant (of the product of the given systems) gives rise to a regular separator~--~provided it can be finitely represented.
We do not even need \WSTS here, but only upward compatibility. 
An inductive invariant is a set of configurations that contains the initial ones, is closed under the transition relation, and is disjoint from the final configurations. 

In a second step, we show that finitely-represented invariants always exist.
To this end, we use ideal completions from lattice theory~\cite{KP92,DBLP:conf/icalp/BlondinFM14,FGII}.
The insight is that, in a \WSTS, any inductive invariant can be finitely represented by its ideal decomposition. 
This ideal decomposition yields states in the ideal completion of the \WSTS, and the first step applies.

The result has theoretical as well as practical applications.
On the theoretical side, recall the following about Petri nets from~\cite{DBLP:conf/asian/MukundKRS98,DBLP:conf/icalp/MukundKRS98}:
Every two Petri net covering languages that are complements of each other are necessarily regular. 
The result not only follows from ours, but the same applies to other classes of \WSTSes, for instance to the languages of LCS, and actually to
\emph{all} \WSTS languages fulfilling the above-mentioned assumptions.
For instance, if the covering language of a Petri net 
is the complement of the language of an LCS, they are necessarily regular;
and if the languages are just disjoint, they are regular separable.

The result is also important in verification.
In 2016 and 2017, the Software Verification Competition was won by  so-called language-theoretic algorithms~\cite{HHP10}.
These algorithms replace the classical state-space search by proofs of language disjointness (between a refinement of the control-flow language and the language of undesirable behavior).
Regular separators are precisely what is needed to prove disjointness.
In this setting, regular separators seem to play the role that inductive invariants play for safety verification~\cite{MP1995}.
Indeed, our results establishes a first link between the two.

We accompany our main result by two more findings.
The first ones are determinization results that broaden the applicability of our results. 
For upward compatibility, we show that every finitely branching \WSTS\ can be determinized.
For downward compatibility, we show that every \WSTS can be determinized if the quasi order is an $\omega^2$-\wqo.
In fact all examples from the literature are $\omega^2$-{\WSTS}, hence they determinize, and in consequence satisfy the assumptions of our results.

Our second accompanying result is on the size of regular separators for Petri nets.
We show how to construct a regular separator in the form of a non-deterministic automaton of size triply exponential in size of the given nets.
With the main result at hand, the result amounts to giving a bound on the size of a finite representation of an inductive invariant.
As inductive invariant, we use the complement of the configurations backward reachable from the final ones.
The estimation starts from a result on the size of a basis for the backward reachable configurations~\cite{DBLP:conf/rp/LazicS15} and reasons about the complementation.
There is a matching lower bound for deterministic automata.

\subparagraph{Outline.} 
Section~\ref{sec:wsts} recalls the basics on \WSTS. 
The determinization results can be found in Section~\ref{sec:expressibility}. 
They prepare the main result in Section~\ref{sec:separability}. 
The state complexity of separators for Petri nets is in Section~\ref{sec:bound}. Section~\ref{sec:remarks} concludes the paper.

\subparagraph{Related Work.}

Separability is a widely-studied problem in Theoretical Computer Science.
A classical result says that every two co-recursively enumerable languages are recursively separable, i.e.~separable by a recursive language~\cite{handbook}.
In the area of formal languages, separability \emph{of} regular languages by subclasses thereof 
was investigated most extensively as a decision problem: Given two regular languages, decide whether they are
separable by a language from a fixed subclass. 
For the following subclasses, among others, the separability problem of regular languages is decidable:
The piecewise-testable languages, shown independently in~\cite{DBLP:conf/icalp/CzerwinskiMM13} and~\cite{DBLP:conf/mfcs/PlaceRZ13},
the locally testable and locally threshold-testable languages~\cite{DBLP:conf/fsttcs/PlaceRZ13}, 
the languages definable in first-order logic~\cite{DBLP:journals/corr/PlaceZ14},
and the languages of certain higher levels of the first-order hierarchy~\cite{DBLP:conf/icalp/PlaceZ14}. 

Regular separability of classes larger than the regular languages attracted little attention until recently. 
As a remarkable example, already in the 70s, the undecidability of regular separability of context-free languages has been shown~\cite{DBLP:journals/siamcomp/SzymanskiW76} (see also a later proof~\cite{DBLP:journals/jacm/Hunt82a}); then the undecidability has been strengthened to visibly pushdown languages~\cite{DBLP:conf/lics/Kopczynski16} and to languages of one-counter automata~\cite{onecounter:sep}.

An intriguing problem, to the best of our knowledge still open, 
is the decidability of regular separability of Petri net languages, under the proviso that acceptance is by \emph{reaching} a distinguished final configuration. 
As for now, positive answers are known only for subclasses of VAS languages: 
$\PSPACE$-completeness for one-counter nets (i.e.~one-dimensional vector addition systems with states)~\cite{onecounter:sep},
and elementary complexity for languages recognizable by Parikh automata (or, equivalently,
by integer vector addition systems)~\cite{DBLP:journals/corr/ClementeCLP16a}.
Finally, regular separability of \emph{commutative closures of} VAS languages has been shown to be decidable in~\cite{DBLP:journals/corr/ClementeCLP16}.
As a consequence of this paper, regular separability of two VAS languages reduces to disjointness of 
the same two VAS languages (and is thus trivially decidable), given that acceptance is by \emph{covering} a distinguished final configuration. 

Languages of upward-compatible \WSTS\ were investigated e.g.~in~\cite{GRV-wsl-07}, where interesting closure properties have been shown, including a natural pumping lemma.
Various subclasses of languages of \WSTS have been considered, e.g.~in~\cite{DELZANNO201312,DBLP:conf/csl/AbdullaDB07,DBLP:conf/apn/Martos-SalgadoR14}.


\section{Well structured transition systems}
\label{sec:wsts}

\subparagraph*{Well Quasi Orders.}

A quasi order $(X, \preceq)$, i.e.~a set $X$ equipped with a reflexive and transitive binary relation $\preceq$, is called 
\emph{well quasi order} (\wqo) if
for every infinite sequence $x_1, x_2, \ldots \in X$ there are indices $i < j$ such that $x_i \preceq x_j$.
It is folklore that $(X, \preceq)$ is \wqo iff it admits neither
an infinite descending sequence (i.e.~it is well-founded) nor an infinite antichain
(i.e.~it has the finite antichain property).

We will be working either with {\wqo}s, or with $\omega^2$-{\wqo}s, a strengthening of {\wqo}s.
We prefer not to provide the technical definition of $\omega^2$-\wqo (which can be found, e.g~in~\cite{Marcone-bqo}), as it would not serve our aims.
Instead, we take the characterization provided by Lemma~\ref{lem:omega2} below as a working definition.
The class of $\omega^2$-{\wqo}s provides a framework underlying the forward WSTS analysis developed in~\cite{FGI,FGII,GEERAERTS2006180}.
Both classes, namely {\wqo}s and $\omega^2$-{\wqo}s, are stable under various operations like taking the Cartesian product, the lifting to finite multisets (multiset embedding), and the lifting to finite sequences (Higman ordering).

A subset $U \subseteq X$ is \emph{upward closed} with respect to $\preceq$
if $u \in U$ and $u' \succeq u$ implies $u' \in U$.
Similarly, one defines \emph{downward closed} sets.
Clearly, $U$ is upward closed iff $X \setminus U$ is downward closed.
The \emph{upward} and \emph{downward closure} of a set $U \subseteq X$ are defined as:
\begin{align*}
    \upclosure{U} = \{x\in X \mid \exists u \in U, \, x \succeq u\}
    \quad \text{ and } \quad
    \downclosure{U} = \{x\in X \mid \exists u \in U, \, x \preceq u\}
    \ .
\end{align*}
%
The family of all upward-closed resp.~downward-closed subsets of $X$ we denote by $\pup{X}$ resp.~$\pdown{X}$.
If $(X, \preceq)$ is a \wqo then every upward closed set is the upward closure of a finite set, namely of the set of its minimal elements.
This is not the case for downward closed set; we thus distinguish a subfamily $\pfindown{X}\subseteq \pdown{X}$ of \emph{finitary} downward closed subsets of $X$, i.e.~downward closures of finite sets.
In general, these are not necessarily finite sets (e.g.~consider the set $\N \cup \{\omega\}$ with $\omega$ bigger than all natural numbers, and the downward closure of $\{\omega\}$).
The set $\pfindown{X}$, ordered by inclusion, is a \wqo whenever $(X, \preceq)$ is:
\begin{restatable}{rlemma}{restateLemmapfin}
\label{c:pfin}
    $\big(\pfindown{X}, {\subseteq}\big)$ is a \wqo iff $(X, \preceq)$ is a \wqo.
\end{restatable}
\noindent
This property does not necessarily extend to the whole set $\pdown{X}$ of all downward closed subsets of $X$.
As shown in~\cite{DBLP:journals/ipl/Jancar99}:
\begin{lemma}
\label{lem:omega2}
    $\big(\pdown{X}, {\subseteq}\big)$ is a \wqo iff $(X, \preceq)$ is an $\omega^2$-\wqo.
\end{lemma}
\noindent
As a matter of fact, \cite{DBLP:journals/ipl/Jancar99} considers the reverse inclusion order on upward closed sets, which is clearly isomorphic to
the inclusion order on downward closed sets.

\subparagraph*{Labeled Transition Systems.}

In the sequel we always fix a finite alphabet $\Sigma$.
A labeled transition system (LTS) $\W = (S, T, I, F)$ over $\Sigma$ consists of a set of \emph{configurations} $S$,
a set of \emph{transitions} $T \subseteq S \times \Sigma \times S$, and subsets $I, F \subseteq S$ of \emph{initial} 
and \emph{final} configurations.
We write $s \trans{a} s'$ instead of $(s, a, s') \in T$. 
A path from configuration $s$ to configuration $s'$ over a word $w = a_0 \cdots a_{k-1}$ 
is a sequence of configurations $s = s_0, s_1, \ldots, s_{k-1}, s_k = s'$
such that $s_i \trans{a_i} s_{i+1}$ for all $i \in \{0, \ldots, k-1\}$. We write $s \trans{w} s'$. 
%
For a subset $X\subseteq S$ of configurations and a word $w\in\Sigma^*$  we write
\begin{align*}
    \reachgen{\W}{X}{w} & = \setof{s \in S}{\existsformula{x \in X}{x \trans{w} s}}
    \ , \\
    \revreachgen{\W}{X}{w} & = \setof{s \in S}{\existsformula{x \in X}{s \trans{w} x}}
\end{align*}
for the set of all configurations reachable (resp.~reversely reachable) from $X$ along $w$.
Note that we have \mbox{$\reachgen{\W}{X}{\varepsilon} = X = \revreachgen{\W}{X}{\varepsilon}$.}
Important special cases will be the set of all $a$-successors (resp.~$a$-predecessors) for $a\in\Sigma$, 
i.e.~configurations reachable along a one-letter word $a$,
and the configurations reachable from the initial configurations $I$ (resp.~reversely reachable from the final configurations $F$):
\begin{align*}
    \succe{\W}{X}{a}  & = \reachgen{\W}{X}{a}
    &
    \reach{\W}{w}  & = \reachgen{\W}{I}{w}
    \\
    \prede{\W}{X}{a}  & = \revreachgen{\W}{X}{a}
    &
    \revreach{\W}{w}  & = \revreachgen{\W}{F}{w}
\end{align*}
satisfying the following equalities for all $w\in\Sigma^*$ and $a\in\Sigma$:
\begin{align}
\label{eq:reachsucc}
    \reach{\W}{w.a} & = \succe{\W}{\reach{\W}{w}}{a}
    \\
\label{eq:reachpred}
    \revreach{\W}{a.w} & = \prede{\W}{\revreach{\W}{w}}{a}
    \ .
\end{align}
We also establish the notation for the whole set of (reversely) reachable configurations:
\begin{align*}\
    \reachall{\W} &= \bigcup_{w\in\Sigma^*} \reach{\W}{w}
    &
    \revreachall{\W} & =  \ \bigcup_{w\in\Sigma^*} \revreach{\W}{w}
    \ .
\end{align*}
An LTS $\W = (S, T, I, F)$ is \emph{finitely branching} if $I$ is finite and for every configuration $s \in S$ and each $a \in \Sigma$ there are only finitely many configurations $s' \in S$ such that $s \trans{a} s'$.
Furthermore, $\W$ is \emph{deterministic} if it has exactly one initial configuration and for every $s\in S$ and each $a\in \Sigma$ there is exactly one $s' \in S$ such that $s \trans{a} s'$.
If $\W$ is deterministic, we write $s' = \succe{\W}{s}{a}$ (resp. $s' = \reach{\W}{w}$) instead of 
$\{s'\} = \succe{\W}{s}{a}$ (resp.~\mbox{$\{s'\} = \reach{\W}{w}$}).

The language recognized by $\W$, denoted $\lang{\W}$, is the set of words
which occur on some path starting in an initial configuration and ending in a final one, i.e.
\[
    \lang{\W} = \setof{w \in \Sigma^* }{ \exists i \in I, f \in F \colon i \trans{w} f }
    \ .
\]
We call two LTS $\W, \W'$  \emph{equivalent} if their languages are the same.
They are \emph{reverse-equivalent} if~\mbox{$\lang{\W} = \setof{\rev(w) }{ w \in \lang{\W'} }$} with \mbox{$\rev(a_1 \ldots a_k) = a_k \ldots a_1$}.

Note that we did not allow for $\varepsilon$-steps in transition systems.
Even if $\varepsilon$-steps can be eliminated by pre-composing and post-composing every transition $s \trans{a} s'$ with the 
reflexive-transitive closure of $\trans{\varepsilon}$,
this transformation does not necessarily preserve finite branching.

\subparagraph{Synchronized Products.}

Consider {LTS} $\W = (S, T, I, F)$  and $\W' = (S', T', I', F')$.
Their \emph{synchronized product} is the LTS \mbox{$\W \times \W' = (\prodS, \prodT, \prodI, \prodF)$} defined as follows:
The configurations are tuples of configurations, $\prodS = S\times S'$, and the initial and final configurations are $\prodI= I\times I'$ and $\prodF = F\times F'$, respectively.
The transition relation is defined by

\vspace*{5pt}
\noindent
\quad
    $
    \begin{tabu}{rl@{}l}
        (s, s') \trans{a} (r, r') \text{ in } \W\times \W'
        &  \text{ if } 
        & \ s \trans{a} r \text{ in } \W
        \\
        & \text{ and }
        &  \ s' \trans{a} r' \text{ in } \W'
        \ .
    \end{tabu}
    $
  
\vspace*{5pt}  
\noindent
It is immediate from the definition that the language of the product is the intersection of the languages, i.e.~\mbox{$\lang{\W \times \W'} = \lang{\W} \cap \lang{\W'}$.}
If $\W$ and $\W'$ both are finitely branching, then so is their product.

\subparagraph*{Upward-Compatible Well-Structured Transition Systems.}

Now we define a labeled version of \emph{well-structured transition systems} as described in~\cite{DBLP:journals/tcs/FinkelS01}, here called upward-compatible well-structured transition system (\UWSTS).
We start by defining the more general notions of quasi ordered {LTS} and \ULTSes.

By a \emph{quasi-ordered LTS} $\W = (S, T, \preceq, I, F)$ we mean an LTS $(S, T, I, F)$ extended with a quasi order $\preceq$ on
configurations. 

An {upward-compatible LTS} (\ULTS) is a quasi-ordered LTS such that the set $F$ of final configurations $F$ is upward closed%
\footnote{Languages defined by upward-closed sets of final configurations are usually called \emph{coverability languages}.} 
with respect to $\preceq$, and the following upward compatibility%
\footnote{In the terminology of~\cite{DBLP:journals/tcs/FinkelS01}, this is strong compatibility.}
is satisfied: 
whenever $s \preceq s'$ and $s \trans{a} r$, then $s' \trans{a} r'$ for some $r' \in S$ such that $r \preceq r'$.
%
%
In other words, $\preceq$ is a simulation relation.
Upward compatibility extends to words:
\begin{restatable}{rlemma}{restateLemmaHigherBetter}
\label{lem:higher-better}
    For $w\in \Sigma^*$, $s \preceq s'$ with $s \trans{w} r$, we have $s' \trans{w} r'$ for some $r' \in S$ with $r \preceq r'$.
\end{restatable}

If the order $(S,\preceq)$ in a \ULTS $\W = (S, T, \preceq, I, F)$ is a \wqo, we call $\W$ a \UWSTS.

As $F$ is upward closed, $\W$ is equivalent to its downward closure $\downclosure{\W}$,
obtained from $\W$ by replacing the set $I$ by its (not necessarily finite) 
downward closure $\downclosure{I}$ with respect to $\preceq$,
and by extending the transition relation as follows: $s \trans{a} r$ in $\downclosure{\W}$ if $s \trans{a} r'$ in $\W$
for some $r' \succeq r$.
Note that with respect to the extended transition relation, $\succe{\downclosure{\W}}{X}{a}$ is downward closed for every $X\subseteq S$.
One easily checks that $\downclosure{\W}$ still satisfies upward compatibility, and
every word accepted by $\W$ is also accepted by $\downclosure{\W}$.
The converse implication follows by the following simulation of $\downclosure{\W}$ by $\W$:
\begin{restatable}{rlemma}{restateLemmaHigherBetterDownwardClosure}
\label{lem:higher-better-downclosure}
    Let $w\in \Sigma^*$.
    Whenever $s \preceq s'$ and $s \trans{w} r$ in $\downclosure{\W}$, then $s' \trans{w} r'$ in $\W$ for some $r' \in S$ such that $r \preceq r'$.
\end{restatable}

The synchronized product of two \ULTSes $(S,T,\preceq,I,F)$ and $(S',T',\preceq',I',F')$ is still a \ULTS with respect to the product order $\prodpreceq$ defined by $(x,x') \prodpreceq (y,y')$ iff $x \preceq y$ and $x' \preceq' y'$.
Indeed, $F\times F'$ is upward closed wrt.~$\prodpreceq$ and the transition relation satisfies upward compatibility.
Since the product order of two {\wqo}s is again a \wqo, the synchronized product of two \UWSTSes is a \UWSTS.

When $\preceq$ is a $\omega^2$-\wqo, the \UWSTS $\W$ is called \emph{$\omega^2$-\UWSTS}.
When the LTS $(S, T, I, F)$ is finitely branching (resp.~deterministic), the \UWSTS $\W$ is called \emph{finitely-branching \UWSTS}
(resp.~\emph{deterministic \UWSTS}).
In the sequel we speak shortly of \UWSTS-languages (resp.~$\omega^2$-\UWSTS-languages, finitely-branching \UWSTS-languages, etc.).

\subparagraph*{Downward-Compatible Well-Structured Transition Systems.}
A downward-compatible well-structured transition system (\DWSTS) is defined like its upward-compatible counterpart, with two modifications.
First, we assume the set of final configurations $F$ to be downward closed, instead of being upward closed.
Second, instead of upward compatibility, we require its symmetric variant, namely \emph{downward compatibility}:
Whenever $s' \preceq s$ and $s \trans{a} r$, then $s' \trans{a} r'$ for some $r' \in S$ such that $r' \preceq r$.
In other words, the inverse of $\preceq$ is a simulation relation.
Downward compatibility extends to words, which can been shown similar to Lemma~\ref{lem:higher-better}.
%
%
Symmetrically to the downward closure of a \UWSTS,
we may define the upward closure $\upclosure{\W}$ of a \DWSTS $\W$ that recognizes the same language.

As above, we also speak of
finitely-branching \DWSTS, or $\omega^2$-\DWSTS. 
We jointly call \UWSTS and \DWSTS just \WSTS.

\subparagraph*{Examples of WSTS.}
Various well known and intensively investigated models of computation happen to be either an \UWSTS or \DWSTS.
The list of natural classes of systems which are \UWSTSes contains, among the others:
vector addition systems (VAS) resp.~Petri nets and their extensions (e.g.~with reset arcs or transfer arcs);
lossy counter machines~\cite{DBLP:conf/stacs/BouajjaniM99};
string rewriting systems based on context-free grammars;
lossy communicating finite state machines (aka lossy channel systems, LCS)~\cite{DBLP:journals/jacm/BrandZ83};
and many others.
In the first two models listed above the configurations are ordered by the multiset embedding, while in the remaining two ones
the configurations are ordered by Higman's subsequence ordering.
The natural examples of \UWSTS, including all models listed above, are $\omega^2$-\UWSTSes and, 
when considered without $\varepsilon$-transitions, finitely-branching.

\DWSTSes are less common.
A natural source of examples is \emph{gainy} models,
like gainy counter system machines or gainy communicating finite state machines.
For an overview, see e.g.~page 31 of~\cite{DBLP:journals/tcs/FinkelS01}.


\section{Expressibility}
\label{sec:expressibility}

Our proof of regular separability assumes one of the \WSTSes to be deterministic.
In this section, we show that this is no strong restriction.
We compare the languages recognized by different classes of \WSTSes, in particular deterministic ones. 
The findings are summarized in Theorem~\ref{thm:summary-expr},  
%
where we use 
$\subseteq$ to say that every language of a \WSTS from one class is also the language
of a \WSTS from another class; and we use $\revsubseteq$ to say that every language of a \WSTS from one class is the
reverse of the language of a \WSTS from another class.

\begin{restatable}{rtheorem}{restateTheoremSummaryExpr}
\label{thm:summary-expr}
    The following relations hold between the \WSTS language classes:
    \begin{align*}
        \text{$\omega^2$-\UWSTS} \ & \subseteq  \text{deterministic \UWSTS}
        =  \text{finitely-branching \UWSTS} \, \subseteq \, \text{all \UWSTS}
        \ ,
        \\
        \text{$\omega^2$-\DWSTS} \ &\subseteq  \text{deterministic \DWSTS}
        \subseteq  \text{finitely-branching \DWSTS} \, = \, \text{all \DWSTS}
        \ ,
        \\
        \text{$\omega^2$-\UWSTS} \ & \revsubseteq \, \text{deterministic \DWSTS}
        \ ,
        \\
        \text{$\omega^2$-\DWSTS} \ & \revsubseteq \,  \text{deterministic \UWSTS}
        \ .
    \end{align*}
\end{restatable}
\noindent
In short, $\omega^2$-\UWSTSes and $\omega^2$-\DWSTSes determinize and reverse-determinize;
finitely-branching \UWSTSes determinize too; and (unrestricted) \DWSTSes are equivalent to finitely-branching \DWSTSes.
In Appendix~\ref{appendix:expressibility}, we formulate and prove a series of lemmata which jointly prove Theorem~\ref{thm:summary-expr}.


\section{Regular Separability}

\label{sec:separability}

We now show our first main results: 
Under mild assumptions, disjoint {\DWSTS} resp.~disjoint {\UWSTS} are regular separable. 
Both theorems follow from a technical result that establishes a surprising link between verification and formal language theory:
Every inductive invariant (of a suitable product \WSTS) that has a finite representation can be turned into a regular separator.
With this, the proofs of regular separability are invariant constructions. 

\subparagraph*{Main Results.}

We say that two languages $\L$ and $\K$ over the same alphabet are \emph{regular separable}
if there is a regular language $\R$ that satisfies $\L \subseteq \R$ and $\R\cap \K = \emptyset$. 
For two {\WSTS} $\W$ and $\W'$, we say that they are regular separable if so are their languages.
Disjointness is clearly necessary for regular separability. 
Our first main results show that for most {\WSTSes} disjointness is also sufficient: 

\begin{restatable}{rtheorem}{restateTheoremSepDWSTS}
\label{thm:sepDWSTS}
    Every two disjoint {\DWSTS}, one deterministic, are regular separable.
\end{restatable}
\begin{restatable}{rtheorem}{restateTheoremSepUWSTS}
\label{thm:sepUWSTS}
    Every two disjoint {\UWSTS}, one finitely branching, are regular separable.
\end{restatable}
\noindent
The results imply that the complement of a non-regular WSTS language cannot be a WSTS language.
They also show that there is no subclass of WSTS languages beyond the regular languages that is closed under complement.
More formally, for a class of languages $\mathcal{C}$, we call a language \emph{doubly $\mathcal{C}$}, if the language as well as its complement are in $\mathcal{C}$.
We obtain the following corollary, generalizing earlier results for Petri net coverability languages~\cite{DBLP:conf/asian/MukundKRS98,DBLP:conf/icalp/MukundKRS98}.
\begin{corollary}
\label{Corollary:Separability}
    (1) Every doubly deterministic \DWSTS language resp.~every doubly finitely-branching \UWSTS language is regular.
    (2) No subclass of finitely-branching \UWSTS languages resp.~deterministic \DWSTS languages beyond REG is closed under complement.
\end{corollary}
The rest of the section is devoted to the proofs.
We will use that the product of two disjoint \WSTS is again a \WSTS with the empty language. 
Whenever the language of a \WSTS is empty, we can find an inductive invariant, a downward-closed set of configurations separating the reachability set from the final configurations. 
Given a finite representation for such an invariant, we show how to turn it into a regular separator, provided one of the \WSTS is deterministic. 
This is our key technical insight, formulated as Theorem~\ref{thm:core} below.

The proof of Theorem~\ref{thm:sepDWSTS} follows directly from this result.
For Theorem~\ref{thm:sepUWSTS}, we consider the ideal completion of an \UWSTS, an extended system in which every downward-closed set has a finite representation.
This in particular applies to inductive invariants, as we show in the form of Proposition~\ref{prop:inducedinvariant}: Any inductive invariant in the original \UWSTS induces an inductive invariant in the ideal completion that has a finite representation.
Combining this result with Theorem~\ref{thm:core} yields the desired proof.

\subparagraph*{Turning Inductive Invariants into Regular Separators.}

Inductive invariants are a standard tool in the safety verification of programs~\cite{MP1995}. 
Technically, an inductive invariant (of a program for a safety property) is a set of program configurations that includes the initial ones, is closed under the transition relation, and is disjoint from the set of undesirable states. 
The following definition lifts the notion to \WSTSes (actually to the more general \ULTSes), where it is natural to require inductive invariants to be downward-closed.

\begin{definition} 
An \emph{inductive invariant} for a \ULTS $\W$ with configurations $S$ is a downward-closed set $X\subseteq S$ with the following three properties:
\begin{align}
    \label{eq:iiI}
    &I \subseteq X
    \ ,
    \\
    \label{eq:iiF}
    &F \cap X = \emptyset
    \ ,
    \\
    \label{eq:iiSucc}
    &\succe{\W}{X}{a} \subseteq X \ \text{for all } a\in\Sigma
    \ .
\end{align}
An inductive invariant $X$ is \emph{finitely-represented} if $X = \downclosure{Q}$ for a finite set $Q \subseteq S$.
\end{definition}
\noindent
By~\eqref{eq:iiI} and~\eqref{eq:iiSucc}, the invariant has to contain the whole reachability set. 
By~\eqref{eq:iiF} and~\eqref{eq:iiSucc}, it has to be disjoint from the predecessors of the final configurations: 
\begin{align*}
    &\reachall{\W} \subseteq X
    \ ,
    &
    &\revreachall{\W} \cap X = \emptyset
    \ .
\end{align*}
This means every inductive invariant shows language emptiness. 
Even more, inductive invariants are complete for proving emptiness, like inductive invariants for programs are (relatively) complete for proving safety~\cite{Cook1978}.

\begin{lemma}
\label{Lemma:Invariants}
    Consider \ULTS $\W$.
    Then $\lang{\W}=\emptyset$ iff there is an inductive invariant for~$\W$.
\end{lemma}
\noindent
For completeness, observe that $X = \downclosure{\reachall{\W}}$ is an inductive invariant.
It is the least one wrt.~inclusion. 
There is also a greatest inductive invariant, namely the complement of $\revreachall{\W}$.  
Note that, due to upward compatibility, $\revreachall{\W}$ is always upward-closed.
  
Other invariants may have the advantage of being easier to represent. 
We will be particularly interested in invariants that are finitely-represented in the sense that they form the downward closure of a finite set.

Here is the core result.
Consider two disjoint \ULTSes. 
Any finitely-represented inductive invariant for the product can be turned into a regular separator.
We will comment on the assumed determinism in a moment.

\begin{theorem}
\label{thm:core}
    Let $\W$ and $\W'$ be disjoint \ULTSes, one of them deterministic, such that $\W \times \W'$ admits a finitely-represented inductive invariant $\downclosure{Q}$. 
    Then $\W$ and $\W'$ are regular separable by the language of a finite automaton with states $Q$.
\end{theorem}
\noindent
For the definition of the separator, let $\W = (S, T, \preceq, I, F)$ be an arbitrary \ULTS and let $\W' = (S', T', \preceq', I', F')$
be a deterministic one such that their languages are disjoint.
Let
\[
    \prodW = \W \times \W' = (\prodS, \prodT, \prodpreceq, \prodI, \prodF)
\]
be their synchronized product.
By the disjointness of $\W$ and $\W'$ we know that $\lang{\prodW} = \emptyset$.
Let $Q \subseteq \prodS$ be a finite set such that $\downclosure{Q}$ is an inductive invariant.

We  define a finite automaton $\A$ with states $Q$ whose language will contain $L(\W)$ while being disjoint from $L(\W')$. 
The idea is to over-approximate the configurations of $\prodW$ by the elements available in $Q$.
The fact that $\reachall{\prodW} \subseteq \downclosure{Q}$ guarantees that every configuration $(s, s')\in \prodS$
has such a representation. 
Since we seek to approximate the language of $\W$, the final states only refer to the $\W$-component. 
Transitions are approximated existentially. 

\begin{definition}
\label{def:sepautomaton}
    We define the \emph{separating automaton induced by $Q$} to be $\A = (Q,\to,Q_I,Q_F)$.
    A state is initial if it dominates some initial configuration of $\prodW$,
    \(
        Q_I = \setof{(s, s') \in Q}{(i, i') \prodpreceq (s, s') \text{ for some } (i, i') \in \prodI}
        \ .
    \)
    As final states we take pairs whose $\W$-component is final,
    \(
        Q_F = \setof{(s, s') \in Q}{ s \in F}
        \ .
    \)
    Finally, the transition relation in $\A$ is an over-approximation of the transition relation in $\prodW$:
    \[
        (s, s') \trans{a} (r, r') \text{ in } \A \quad \text{ if } (s, s') \trans{a} (t, t') \text{ in } \prodW \text{ for some } (t, t') \prodpreceq (r, r')\ .
    \]
\end{definition}
Figure~\ref{Figure:SepTransRel} illustrates the construction.

\begin{figure}[b]
    \begin{align*}
    \xymatrix
    {
        && (r,r') \in Q \\
        Q \ni (s,s') \ar@/^1pc/@{.>}[rru]^a_{\text{ in } \A} \ar[rr]^a_{\text{ in } \prodW} && 
        (t,t') \in \prodS \ar@{}[u]|{\rot{\prodpreceq} \qquad}
    }
    \end{align*}
    \caption{The transition relation of $\A$.}
    \label{Figure:SepTransRel}
\end{figure}

To show separation, we need to prove $\lang{\W} \subseteq \lang{\A}$ and \mbox{$\lang{\A}\cap \lang{\W'} = \emptyset$}. 
We begin with the former. 
As $\W'$ is deterministic, $\prodW$ contains all computations of $\W$. 
Due to upward compatibility, $\A$ over-approximates the computations in $\prodW$.
Combining these two insights, which are summarized in the next lemma, yields the result.

\begin{restatable}{rlemma}{restateLemmaSim}
\label{lem:sim} 
    (1) For every $s \in \reach{\W}{w}$ there is some \mbox{$(s, s') \in \reach{\prodW}{w}$}.
    (2) For every $(s, s') \in \reach{\prodW}{w}$ there is some $(r, r') \in \reach{\A}{w}$ with $(s, s') \prodpreceq (r, r')$.
\end{restatable}
\begin{proposition}
\label{prop:SepW}
    $\lang{\W} \subseteq \lang{\A}$\ .
\end{proposition}

It remains to prove disjointness of $\lang{\A}$ and $\lang{\W'}$. 
The key observation is that, due to determinism, $\W'$ simulates 
the computations of $\A$ --- in the following sense:
If upon reading a word $\A$ reaches a state $(s, s')$, then the unique computation of $\W'$ will reach a configuration dominated by $s'$.

\begin{restatable}{rlemma}{restateLemmaInvsimWprime}
\label{lem:invsimWprime}
    For every $w \in \Sigma^*$ and every $(s, s') \in \reach{\A}{w}$ we have $\reach{\W'}{w} \preceq' s'$.
\end{restatable}
\noindent
With this lemma we can show disjointness.
Towards a contradiction, suppose some word $w$ satisfies $w\in \lang{\A} \cap \lang{\W'}$.
As $w\in \lang{\A}$, there is a configuration $(s, s') \in \reach{\A}{w}$ with $s \in F$.
As $w\in \lang{\W'}$, the unique configuration $\reach{\W'}{w}$ belongs to $F'$. 
With the previous lemma and the fact that $F'$ is upward-closed, we conclude $s'\in F'$.
Together, $(s, s')\in \prodF$, which contradicts the fact that $\downclosure{Q}$ is an inductive invariant, Property~\eqref{eq:iiF}.

\begin{proposition}
\label{prop:SepWprime}
    $\lang{\A} \cap \lang{\W'} = \emptyset$\ .
\end{proposition}
\noindent
Together, Proposition~\ref{prop:SepW} and~\ref{prop:SepWprime} show Theorem~\ref{thm:core}.
With Theorem~\ref{thm:core} at hand, the proof of regular separability for \DWSTSes follows easily.

\begin{proof}[Proof of Theorem~\ref{thm:sepDWSTS}]
    Consider an arbitrary \DWSTS $\W = (S, T, \preceq, I, F)$ and a deterministic one $\W' = (S', T', \preceq', I', F')$.
    We start with the observation that the inversed versions of $\W$ and $\W'$, namely with the orders $\preceq^{-1}$ and 
    $(\preceq')^{-1}$ and denoted by $\W^{-1}$ and $(\W')^{-1}$, are \ULTSes.
    We claim that these \ULTSes satisfy the assumptions of Theorem~\ref{thm:core}.
    The language of \mbox{$\prodW^{-1}=\W^{-1} \times (\W')^{-1}$} is empty since the language of $\prodW=\W \times \W'$ is empty and inversion does not change the language, $\lang{\W} = \lang{\W^{-1}}$ and similar for $\W'$. 
    Inversion also does not influence determinism. 
    
    It remains to find an inductive invariant of $\prodW^{-1}$ that is finitely represented.
    We claim that   
    \(
    X\ =\ \downclosureinv{\reachall{\prodW^{-1}}}
    \)
    is a suitable choice.
    The subscript indicates that the downward closure is computed relative to the quasi order of $\prodW^{-1}$.
    As the language of $\prodW^{-1}$ is empty, $X$ is an inductive invariant by Lemma~\ref{Lemma:Invariants}.
    For the finite representation, note that inversion does not change the transition relation. 
    Hence, $\prodW$ and $\prodW^{-1}$ reach the same configurations, 
    \(
    \reachall{\prodW^{-1}}=\reachall{\prodW}=Z\ .
    \)   
    With the definition of inversion, $X=\downclosureinv{Z} = \upclosure{Z}$ holds. 
    Moreover, $\upclosure{Z} = \upclosure{\min(Z)}$, with minimum and upward closure computed relative to $\prodW$. 
    Since the configurations of $\prodW$ are well quasi ordered, $\min(Z)$ is finite. 
    Another application of inversion yields $X=\upclosure{\min(Z)}=\downclosureinv{\min(Z)}$.  
    Hence, $X$ is a finitely-represented downward-closed subset of $\prodW^{-1}$. 
    
    By Theorem~\ref{thm:core}, the languages of $\W^{-1}$ and $(\W')^{-1}$ are regular separable and so are the languages of $\W$ and $\W'$. 
\end{proof}

\subparagraph*{Ideal Completions of \UWSTS.}
The proof of regular separability for \UWSTS is more involved. 
Here, we need the notion of ideal completions~\cite{DBLP:conf/icalp/BlondinFM14,FGII}.
We show that any invariant for a \WSTS yields a 
finitely-represented invariant for the corresponding ideal completion. 
Theorem~\ref{thm:sepUWSTS} follows from this. 

An \emph{ideal} in a  \wqo $(X, \preceq)$ is a non-empty downward-closed subset $Z\subseteq X$ which is directed:
For every $z, z'\in Z$ there is a $z'' \in Z$ with $z  \preceq z''$ and $z' \preceq z''$.  
Every downward-closed set decomposes into finitely many ideals. 
In fact, the finite antichain property is sufficient and necessary for this. 

\begin{lemma}[\cite{KP92,FGII,demystifying}]
\label{lem:idec}
    In a \wqo, every downward-closed set is a finite union of ideals.
\end{lemma}
\noindent
We use $\idec{X}{Z}$ to denote the set of inclusion-maximal ideals in $Z$.
By the above lemma, $\idec{X}{Z}$ is always finite and 
\begin{align}
\label{eq:union}
    Z = \bigcup \idec{X}{Z}\ .
\end{align}
\noindent
We will also make use of the fact that ideals are irreducible in the following sense.

\begin{lemma}[\cite{KP92,FGII,demystifying}]
\label{lem:idealsirred}
    Let $(X, \preceq)$ be a \wqo. 
    If $Z \subseteq X$ is downward-closed and $I \subseteq Z$ is an ideal, then
    $I \subseteq J$ for some $J \in \idec{X}{Z}$.
\end{lemma}

The \emph{ideal completion} $(\barX, \subseteq)$ of $(X, \preceq)$ has as elements all ideals
in $X$.
The order is inclusion. 
The ideal completion $\barX$ can be seen as extension of $X$; indeed, every element $x\in X$ is represented
by $\downclosure{\{x\}} \in \barX$, and inclusion among such representations
coincides with the original quasi order~$\preceq$.
Later, we will also need general ideals that may not be the downward closure of a single element.

In \cite{FGII,DBLP:conf/icalp/BlondinFM14}, the notion has been lifted to \WSTS\ $\W=(S, T, \preceq, I, F)$.
The ideal completion of $\W$ is the \ULTS $\idcompl{\W}$, where the given $\wqo$ is replaced by its ideal completion. 
The initial configurations are the ideals in the decomposition of $\downclosure{I}$.  
The transition relation is defined similarly, by decomposing $\downclosure{\succe{\W}{X}{a}}$, with $X$ an ideal. 
The final configurations are the ideals that intersect $F$. 

\begin{definition}[\cite{FGII,DBLP:conf/icalp/BlondinFM14}]
    For an \UWSTS $\W = (S, T, \preceq, I, F)$, we define its \emph{ideal completion} 
    $\idcompl{\W} = (\barS, \barT, \subseteq, \barI, \barF)$, where 
    $(\barS, \subseteq)$ is the ideal completion of $(S, \preceq)$, 
    the transition relation is defined by
    \(
        \succe{\idcompl{\W}}{X}{a} \ = \ \idec{S}{\downclosure{\succe{\W}{X}{a}}}\ ,
    \)
    $\barI = \idec{S}{\downclosure{I}}$, and 
    $\barF = \setof{X\in\barS}{X\cap F \neq \emptyset}$. 
\end{definition}
\noindent
Using upward compatibility in $\W$, language equivalence holds and determinism is preserved.

\begin{restatable}{rlemma}{restateLemmaIdcompl}
\label{lem:idcompl}
  The ideal completion $\idcompl{\W}$ of an \UWSTS $\W$ is a \ULTS.
  We have $\lang{\idcompl{\W}}=\lang{\W}$. 
  If $\W$ is deterministic, then so is $\idcompl{\W}$.
\end{restatable}
\noindent
As a matter of fact, $\idcompl{\W}$ is even finitely branching, but we do not need this property.

The purpose of using ideal completions is to make it easier to find inductive invariants that are finitely represented. 
Assume the given \UWSTS $\W$ has an inductive invariant $X$, not necessarily finitely represented. 
By definition, $X$ is downward-closed.
Thus, by Lemma~\ref{lem:idec}, $X$ is a \emph{finite} union of ideals.
These ideals are configurations of the ideal completion $\idcompl{\W}$. 
To turn $\idec{S}{X}$ into an inductive invariant of $\idcompl{\W}$, it remains to take the downward closure of the set. 
As the order among ideals is inclusion, this does not add configurations.
In short, an inductive invariant for $\W$ induces a finitely-represented inductive invariant for $\idcompl{\W}$.

\begin{restatable}{rproposition}{restatePropositionInducedinvariant}
\label{prop:inducedinvariant}
    If $X \subseteq S$ is an inductive invariant of $\W$, 
    $\downclosure{\idec{S}{X}}$ is a finitely-represented inductive invariant of $\idcompl{\W}$.
\end{restatable}
\begin{proof}
    Define $Q = \idec{S}{X}$.
    Since $Q$ contains all ideals $Y \subseteq X$ that are maximal wrt.~inclusion, $\downclosure{Q}$ contains all ideals $Y \subseteq X$.
    We observe that
    \[
    X \stackrel{\eqref{eq:union}}{=} \bigcup Q = \bigcup \downclosure{Q}
    \ . 
    \]
    By Lemma~\ref{lem:idec}, $Q$ is finite and thus $\downclosure{Q}$ is finitely-represented.
    It remains to check that $\downclosure{Q}$ satisfies the Properties~\eqref{eq:iiI}, \eqref{eq:iiF}, and~\eqref{eq:iiSucc}.
    
    To show Property~\eqref{eq:iiI}, we need to prove $\idec{S}{\downclosure{I}} \subseteq \downclosure{Q}$.
    We have $I \subseteq X$ by Property~\eqref{eq:iiI}, and since $X$ is downward-closed, we obtain $\downclosure{I} \subseteq X$.
    Consequently, any ideal that is a subset of $\downclosure{I}$ is also a subset of $X$, and $\downclosure{Q}$ contains all such ideals.
    
    For Property~\eqref{eq:iiF}, assume towards a contradiction that $\downclosure{Q}$ contains an ideal $Y$ that is final in $\idcompl{\W}$.
    By definition, this means $Y$ contains a final configuration.
    Since $Y \subseteq X$, we obtain a contradiction to $X \cap F = \emptyset$, Property~\eqref{eq:iiF}.

    To check the inclusion $\succe{\idcompl{\W}}{\downclosure{Q}}{a} \subseteq \downclosure{Q}$, we pick an ideal $Y \in \downclosure{Q}$ and show $\succe{\idcompl{\W}}{Y}{a} \subseteq \downclosure{Q}$.
    Recall the definition $\succe{\idcompl{\W}}{Y}{a} = \idec{S}{\downclosure{\succe{\W}{Y}{a}}}$.
    Thus, any element of $\succe{\idcompl{\W}}{Y}{a}$ is an ideal that is a subset of $\downclosure{\succe{\W}{Y}{a}}$.
    We have $\succe{\W}{X}{a} \subseteq X$ by Property~\eqref{eq:iiSucc}.
    This implies $\succe{\W}{Y}{a} \subseteq X$ as $Y \subseteq X$, and even \mbox{$\downclosure{\succe{\W}{Y}{a}} \subseteq X$} as $X$ is downward-closed.
    Hence, any ideal that is a subset of $\downclosure{\succe{\W}{Y}{a}}$ is also subset of $X$, and thus an element of
    $\downclosure{Q}$.
\end{proof}
Theorem~\ref{thm:core} expects invariants for \UWSTS of a particular shape, namely products $\W\times \W'$. 
We now show that the operation of ideal completion commutes with taking products of \UWSTS, a fact that will be key to the proof of Theorem~\ref{thm:sepUWSTS}. 
We start by recalling that the ideals in a product \wqo $X\times Y$ are precisely the products of the ideals in $X$ and in $Y$. 

\begin{lemma}[\cite{KP92,FGII,demystifying}]
\label{lem:prodideals}
    A set $Z\subseteq X\times Y$ is an ideal iff $Z = I \times J$, where $I \subseteq X$ 
    and $J\subseteq Y$ are ideals.
\end{lemma}
\noindent
Lemma~\ref{lem:prodideals} yields the mentioned commutativity.

\begin{restatable}{rlemma}{restateLemmaIsoprod}
\label{lem:isoprod}
    For two {\UWSTS}es $\W$ and $\W'$,
    $\idcompl{\W} \times \idcompl{\W'}$ and $\idcompl{\W \times \W'}$ are isomorphic.
\end{restatable}
\noindent
We are now prepared to apply Theorem~\ref{thm:core} once more to establish our second main result. 

\begin{proof}[Proof of Theorem~\ref{thm:sepUWSTS}]
    Let $\W = (S, T, \preceq, I, F)$ and  $\W' = (S', T', \preceq', I', F')$ be disjoint {\UWSTS} and $\W'$ finitely branching.
    By Theorem~\ref{thm:summary-expr} 
    we can assume $\W'$ is deterministic.
    
    We would like to construct a finitely-represented inductive invariant in the synchronized product of the ideal completions $\idcompl{\W} \times \idcompl{\W'}$ and then apply Theorem~\ref{thm:core}. 
    Indeed, by Lemma~\ref{lem:idcompl} we know that the ideal completions are disjoint \ULTSes, and that the latter one is still deterministic, so they satisfy the assumptions. 

    Relying on Lemma~\ref{lem:isoprod} we prefer to show the existence of a finitely-represented inductive invariant in 
    \(
        \idcompl{\W\times \W'}.
    \)
    Using Proposition~\ref{prop:inducedinvariant}, it is sufficient to find any inductive invariant in $\W\times \W'$, it does not have to be finitely-represented. 
    We know that such an inductive invariant exists by Lemma~\ref{Lemma:Invariants}, since we assume  $\lang{\W\times \W'} = \lang{\W} \cap \lang{\W'} = \emptyset$.
\end{proof}

\subparagraph*{Effective Representation.}
The states of the separating automaton in the proof of Theorem~\ref{thm:sepUWSTS} are ideals in the product systems.
With Lemma~\ref{lem:prodideals}, these are tuples of ideals in the original systems.
For most types of \UWSTSes, it is known how ideals can be effectively represented, i.e.~how to obtain finite representations on which the successors can be computed.
We briefly mention such a construction for Petri nets in Lemma~\ref{lem:natideals}, see e.g.~\cite{DBLP:conf/icalp/BlondinFM14} for more examples.
In general, one may exploit the fact that ideals are downward-closed sets, which in turn are complements of upward-closed sets that can be represented by finitely many minimal elements -- an idea first proposed in \cite{GRVB06}.
Note that in the proof of Theorem~\ref{thm:sepUWSTS}, we invoke Theorem~\ref{thm:summary-expr} to determinize the given finitely-branching \UWSTS.
The states of the resulting \UWSTS are finitary downward-closed sets of states of the original one.
For most types of \UWSTSes, this construction can be avoided.
We demonstrate this for the case of Petri nets in the proof of Proposition~\ref{Proposition:GeneralUpper}.


\section{Separator Size: The Case of Petri Nets}
\label{sec:bound}

The \UWSTS associated to Petri nets are finitely branching. 
Hence, Theorem~\ref{thm:sepUWSTS} applies:
Whenever the coverability languages of two Petri nets are disjoint, they are regular separable.
We now show how to construct a triply-exponential non-deterministic finite automaton (NFA) separating two such languages, provided they are disjoint.
Moreover, for deterministic finite automata (DFA), we show that this size cannot be avoided.

\begin{theorem}\label{Theorem:UpperPN}
    Let $\lang{N_1}$, $\lang{N_2}$ be disjoint Petri net coverability languages.
    There is an NFA $\A$ of size triply exponential in $\card{N_1}+\card{N_2}$ such that $\lang{A}$ separates $\lang{N_1}$ and $\lang{N_2}$.
\end{theorem}

\begin{theorem}\label{Theorem:LowerPN}
    In general, Petri net coverability languages cannot be separated by DFA of less than triply-exponential size.
\end{theorem}
\noindent
Instead of invoking Theorem~\ref{thm:sepUWSTS}, which uses Theorem~\ref{thm:summary-expr} to determinize, we directly show how to construct an equivalent instance of the separability problem in which one of the nets is deterministic.
In this setting, we prove an upper bound that combines Theorem~\ref{thm:core} with a size estimation for an ideal decomposition. 
We then show how to handle non-determinism.
%
%
%
The lower bound combines a classical result from automata theory, showing that minimal DFA may have exponentially many states~\cite{Kozen1997}, with a Petri net construction due to Lipton~\cite{Lipton}.

\subparagraph*{Petri Nets.}
%
A \emph{Petri net} over the alphabet~$\Sigma$ is a tuple \mbox{$N = (P,T,F, \lambda, M_0, M_f)$} where $P$ is a finite set of places, $T$ is a finite set of transitions with $P \cap T = \emptyset$, $F \colon (P \cup T) \times (P \cup T) \to \N$ is a flow function, and $\lambda \colon T \to \Sigma$ is a labeling of the transitions. 
The runtime behavior of Petri nets is defined in terms of so-called \emph{markings} from  $M\in \N^d$ with $d=\card{P}$.
If $M(p)=k>0$, we say place $p$ carries $k$ tokens.
We assume to be given an initial and a final marking, $M_0, M_f\in \N^d$. 
Markings are changed by firing transitions:
A transition $t\in T$ is \emph{enabled} in marking $M\in\N^d$, if $M(p) \geq F(p,t)$ for all places $p$. 
An enabled transition can be \emph{fired} leading to the marking $M'$ with 
$M'(p) = M(p) - F(p,t) + F(t,p)$,  denoted $M \fire{t} M'$.
Note that enabledness and firing are upward compatible with the componentwise ordering $\leq$ on markings, in the following sense. 
If $M_1\leq M_2$ and $M_1\fire{t}M_1'$, then $M_2\fire{t}M_2'$ with $M_1'\leq M_2'$.

Relying on this compatibility, we can define the \emph{\UWSTS induced by $N$} to be 
$\WSTSN = (\Nat^P, T', \leq, \set{M_0}, \upclosure{M_f})$.
The transition relation is defined by $(M,a,M') \in T'$ if there is a transition $t\in T$  such that $M \fire{t} M'$ and $\lambda(t) = a$. 
The language of $\WSTSN$ is also called the \emph{(coverability) language of $N$}\footnote{We consider covering the final marking as acceptance condition, i.e.~a sequence of transitions is accepting if it reaches some marking $M'$ with $M'(p) \geq M_f(p)$ for all $p \in P$.}, and denoted by $\lang{N}$.  
We call $N$ \emph{deterministic} if $\WSTSN$ is. 

We use a \emph{product operation} on Petri nets $N_i = (P_i,T_i,F_i, \lambda_i, M_{0, i}, M_{f, i})$, $i=1, 2$.
The product Petri net is obtained by putting the places of $N_1$ and $N_2$ side by side and creating a new transition for all pairs of transitions in $T_1\times T_2$ that carry the same label. 
Formally, $N_1\times N_2 = (P, T, F, \lambda, M_0, M_f)$ with $P = P_1 \dotcup P_2$,
$T = \setof{ (t_1, t_2) \in T_1\times T_2}{ \lambda(t_1) = \lambda(t_2) }$.
We have $\lambda(t_1, t_2)= \lambda(t_1) = \lambda(t_2)$. 
The flow function is defined by the flow functions of the component Petri nets,
$F ( p, (t_1,t_2) ) = F_x( p, t_x)$  and 
$F ( (t_1,t_2), p ) = F_x( t_x, p)$, where $x = i$ if $p \in P_i$.
We have $M_0(p) = M_{0, i} (p)$ for $p \in P_i$, and similar for $M_f$.
The product operation on Petri nets coincides with the product on \UWSTS. 

\begin{lemma}\label{Lemma:ProductNet}
    $\wstsof{N_1\times N_2}$ is isomorphic to $\wstsof{N_1}\times \wstsof{N_2}$.
\end{lemma}

We will need the size of a Petri net. 
It is defined using a binary encoding of the values in the range of the flow function and in the markings. 
Define the \emph{infinity norm} of a vector $M \in \N^d$ to be 
$\norm{M} = \max_{p\in P} {M(p)}$. 
We extend this notion to matrices, sets of vectors, and functions by taking the maximum over all entries, elements, and elements in the range, respectively. 
The \emph{size} of the Petri net $N$ is now
\(
    \card{N}= \card{P} \card{T} (1+ \lceil\log_2(1+\norm{F})\rceil ) + \card{M_0}+|M_f|\ .
\)
The size of a marking $M$ is $\card{M}=  |P| ( 1+ \lceil\log_2(1+\norm{M})\rceil )$.


\subparagraph*{An Upper Bound Assuming Determinism.}
Theorem~\ref{thm:core} assumes that one of the \UWSTS is deterministic. 
We now show that for Petri nets, in this case, the regular separator is (an NFA of size) at most doubly exponential in the size of the input Petri nets. 
%

To prove the result, we show how a size estimation for the basis of $\revreachall{\W}$ with $\W=\wstsof{N_1\times N_2}$ can be turned into a size estimation for the ideal decomposition of the complement.
The size estimation of the basis is the following result.
It is obtained by inspecting Abdulla's backward search~\cite{ACJT96}. 

\begin{theorem}[Bozzelli \& Ganty~\cite{DBLP:conf/rp/BozzelliG11}]
    \label{thm:BozelliGanty}
    Consider a Petri net $N$ with final marking $M_f$. 
    Then
    $\revreachall{\WSTSN}=\upclosure{\set{v_1, \ldots, v_k}}$, where $k$ as well as $\norm{\set{v_1, \ldots, v_k}}$ are bounded from above by
    \[
    g =
    \left(
    \card{T} \cdot
    \left(
    \norm{F}
    +
    \norm{M_0}
    +
    \norm{M_f}
    +
    2
    \right)
    \right)^{2^{\bigO{\card{P} \cdot \log \card{P}}}}
    \ .
    \]
\end{theorem}
\noindent
By Lemma~\ref{Lemma:Invariants}, $\Nat^d \setminus \revreachall{\W}$ is an inductive invariant of $\W$ (provided the language is empty).
We can now apply Lemma~\ref{lem:idec} to finitely represent this set by its ideal decomposition.
To represent this ideal decomposition in turn, we have to explicitly represent  ideals in $\N^d$. 
The following lemma gives such a representation.

Let $\Nat_\omega$ denote $\Nat$ extended by a new top element~$\omega$. 
Every ideal in $\Nat^d$ is the downward closure $\downclosure{u}$ of a single vector $u\in \Nat_\omega^d$.
The lemma moreover shows how to compute the intersection of two ideals and how to obtain the ideal decomposition of the complement $\Nat^d \setminus \upclosure{v}$ of the upward closure of a vector $v \in \Nat^d$.
%
%
\begin{lemma}[see e.g.~\cite{DBLP:conf/rp/LazicS15}]
\label{lem:natideals}
    (1)~The ideals in $\Nat^d$ have the shape $\downclosure{u}$ for $u \in \Nat_\omega^d$.
    (2)~For two ideals $\downclosure{u_1}, \downclosure{u_2}$ of $\Nat^d$, the intersection is  
    $\downclosure{u_1}\cap\downclosure{u_2}=\downclosure{u}$ with $u(i) = \min \set{ u_1(i), u_2(i)}$. 
    (3)~For $v \in \Nat^d$, we have 
    \(
        \idec{}{\Nat^d \setminus \upclosure{v} } = \setof{ \downclosure{u_{<v(j)}} }{ j \in \interval{d} }
        \ , 
    \)
    where $u_{<v(j)}(j)= v(j)-1$ and $u_{<v(j)}(i)=\omega$ for $i\neq j$.
\end{lemma}
We can now combine Theorem~\ref{thm:BozelliGanty} and Lemma~\ref{lem:natideals} to obtained our upper bound.

\begin{restatable}{rproposition}{restatePropositionUpperbounddet}
    \label{prop:upperbounddet}
    Let $N_1$ be an arbitrary Petri net and let $N_2$ be deterministic.
    If $N_1$ and $N_2$ are disjoint, they can be separated by an NFA of size doubly exponential in $\card{N_1}+\card{N_2}$.
\end{restatable}

\subparagraph*{A General Upper Bound.}

The previous result yields a doubly-exponential separator in the case where $N_2$ is deterministic. 
We now show how to get rid of this assumption and construct a separator in the general case. 

\begin{proposition}
\label{Proposition:GeneralUpper}
    Let $N_1$ and $N_2$ be disjoint Petri nets.
    Then they are separable by an NFA of size triply exponential in $\card{N_1}+\card{N_2}$.
\end{proposition}
\noindent
The proof transforms $\NA$ and $\NB$ into $\NAlambda$ and $\NBfree$ so that $\NBfree$ is deterministic, invokes Proposition~\ref{prop:upperbounddet}, and then turns the resulting separator for $\NAlambda$ and $\NBfree$ into a separator for $\NA$ and $\NB$. 
The approach is inspired by~\cite{DBLP:journals/corr/ClementeCLP16a}.

Let $\NB$ be non-deterministic with labeling function $\lambda \colon \TB \to \Sigma$. 
We define $\NBfree$ to be a variant of $\NB$ that is labeled by the identity function, i.e.~$\NBfree$ is a Petri net over the alphabet $\TB$. 
We have $\lang{\NB} = \lambda(\lang{\NBfree})$, where we see $\lambda$ as a homomorphism on words. 
We furthermore define $\NAlambda$ to be the $\TB$-labeled Petri net obtained from $\NA$ as follows. 
For each $a$-labeled transition $t_1$ of $\NA$ and each $a$-labeled transition $t$ of $\NB$, $\NAlambda$ contains a $t$-labeled copy $t_{1}^{t}$ of $t_1$ with the same input-output behavior.
Transition $t_1$ itself is removed. 


\begin{restatable}{rlemma}{restateLemmaPNDeterminization}
\label{lem:PNDeterminization}
    $\lang{\NA \times \NB} = \lambda ( \lang{ \NAlambda \times \NBfree })$.
\end{restatable}
\noindent
With this lemma, and since $\NA$ and $\NB$ are disjoint, $\NAlambda$ and $\NBfree$ have to be disjoint.
As $\NBfree$ is deterministic, we can apply Proposition~\ref{prop:upperbounddet} 
and obtain a separator for $\NAlambda$ and $\NBfree$. 
Let $\A$ be the doubly-exponential NFA over the alphabet $\TB$ with $\lang{\NAlambda} \subseteq \lang{\A}$ and $\lang{\NBfree} \cap \lang{\A} = \emptyset$.
We show how to turn $\A$ into a separator for $\NA$ and $\NB$.
The first step is to determine the complement automaton $\complA$, which satisfies
$\lang{\NBfree} \subseteq \lang{\complA}$ and $\lang{\NAlambda} \cap \lang{\complA} = \emptyset$. 
The second step is to apply $\lambda$ to $\complA$.
Let $\B=\lambda(\complA)$ be the automaton obtained from $\complA$ by relabeling each $t$-labeled transition to $\lambda(t)$.
The following lemma shows that $\B$ is a separator for the original nets.
The observation that the size of $\complA$ and hence the size of $\B$ is at most  exponential in the size of $\A$ concludes the proof of Proposition~\ref{Proposition:GeneralUpper}

\begin{restatable}{rlemma}{restateLemmaBisSep}
\label{lem:BisSep}
    $\lang{\NB} \subseteq \lang{\B}$ and $\lang{\NA} \cap \lang{\B} = \emptyset$.
\end{restatable}
\noindent
Note that $\lambda(\A)$ is not necessarily a separator:
There might be $u \in \lang{A}$, $u \not\in \lang{\NBfree}$ such that there is $u' \in \lang{\NBfree}$ with $\lambda(u) = \lambda(u')$.
Thus, $\lambda(u) \in \lambda(\lang{\A}) \cap \lang{\NB}$.

\subparagraph{A Lower Bound.}

We now consider separation by \emph{deterministic} finite automata (DFA). 
In this case, we can show a triply-exponential lower bound on the size of the separator.

\begin{restatable}{rproposition}{restatePropositionLowerbound}
\label{prop:lowerbound}
    For all $n \in \N$, there are disjoint Petri~nets $N_0(n)$ and $N_1(n)$ of size polynomial in $n$ such that any separating DFA has size at least triply exponential in $n$.
\end{restatable}
\noindent
Our proof relies on the classical result that for each $x \in \binary$ and each $k \in \N$, the minimal DFA for the language
\(
    \mathcal{L}_{x@k} = \setof{ w \in \binary^{\geq k}   }{ \text{ the } k\text{-last letter in } w \text{ is } x}
\)
needs at least $2^k$ states~\cite{Kozen1997}. 
To obtain the desired lower bound, we will show how to generate $\mathcal{L}_{x@k}$ for a doubly-exponential number $k$ by a polynomially-sized Petri net.
To this end, we make use of Lipton's proof of $\EXPSPACE$-hardness for coverability~\cite{Lipton}.




\section{Conclusion}
\label{sec:remarks}

We have shown that, under mild assumptions, disjointness of \WSTS languages implies their regular separability. 
In particular, we have shown that if one of two disjoint upward-compatible \WSTS is finitely branching, they are regular separable.
Using our expressibility results, it is also sufficient if the underlying order for one of the two is an $\omega^2$-\wqo.
A similar result holds for downward-compatible \WSTS assuming that one of them is deterministic or the underlying order is an $\omega^2$-\wqo.
As \WSTSes are typically $\omega^2$-\WSTSes
, our result already implies the decidability of regular separability for almost all \WSTSes of practical relevance.

Our work brings together research on inductive invariants and regular separability.
We show that a finite representation of an inductive invariant for the product system can be transformed into a regular separator.
For Petri nets, one may use any representation of the coverability set.
As we show, it is  beneficial in terms of the worst-case size, to use an inductive invariant obtained from the backward coverability algorithm~\cite{ACJT96}.
For lossy channel systems, the coverability set is not computable~\cite{Mayr03}, but one can obtain a finitely-represented inductive invariant e.g.~from the EEC-algorithm~\cite{GEERAERTS2006180}.

We leave some questions without answer.
It is not clear whether the assumptions of Theorems~\ref{thm:sepUWSTS} and~\ref{thm:sepDWSTS} 
are necessary; we were neither able to drop the assumptions, nor to provide a counterexample. 
Similarly, we do not know whether the inclusions in Theorem~\ref{thm:summary-expr} are strict.
Finally, in the case of Petri nets, closing the gap between the triply-exponential size of the NFA separator 
and the triply-exponential lower bound for DFA remains an open problem.

As future work, one could consider the \emph{well-behaved transition systems (WBTS)} of~\cite{DBLP:journals/corr/BlondinFM16}, a  generalization of \WSTS where only the finite-antichain property is required.

\newpage
\bibliographystyle{plainurl}
\bibliography{35clmmns}

\newpage
\appendix

\section{Details for Section \ref{sec:wsts}}

\restateLemmapfin*

\begin{proof}
	We first assume that  $\big(\pfindown{X}, {\subseteq}\big)$ is a \wqo and prove that $(X, \preceq)$ is a \wqo.
    Consider any infinite sequence $x_1,x_2 \dots$, where each $x_i \in X$.
    By our assumption, in the infinite sequence
    $\downclosure{\set{x_1}}, \downclosure{\set{x_2}},\dots $
    we find $i < j$ such that $\downclosure{\set{x_i}} \subseteq \downclosure{\set{x_j}}$.
    We conclude $x_i \preceq x_j$ as desired.
    
    For the other direction, we will assume that $(X, \preceq)$ is a \wqo.
    Using Higmann's lemma, also $(X^*,\preceq^*)$ is a \wqo.
    Here $X^*$ is set of all finite sequences over $X$ and $\preceq^*$ is the subsequence order, i.e.~$w \preceq^* v$ if $w$ is obtained from $v$ by deleting symbols and/or replacing them by smaller symbols with respect to $\preceq$.
    Now consider any infinite sequence	$X_1,X_2,\dots$ in $\pfindown{X}$.
    By definition, each $X_i$ can be written as $\downclosure{\set{ u^i_1, \dots ,u^i_{n_i} }}$
    for appropriately chosen $u^i_j$.
    We represent each $X_i$ by $u^i_1 \dots u^i_{n_1} \in X^*$ and consider the sequence
    \[
        u^1_1 \dots u^1_{n_1}, u^2_1 \dots u^2_{n_2},\ldots
    \]
    in $X^*$.
    Using the fact that $(X^*,\preceq^*)$ is a \wqo, we obtain $i < j$ such that
    $u^i_1 \dots u^i_{n_i} \preceq^* u^j_1 \dots u^j_{n_j}$.
    From this, we immediately obtain $X_i \subseteq X_j$.
\end{proof}

\restateLemmaHigherBetter*

\begin{proof}
	We claim that for all $w \in \Sigma^*$, $s \preceq s'$ and $s \trans{w} r$ implies $s' \trans{w} r'$ for some $r' \in S$ with $r \preceq r'$.
    
    We proceed by induction on $w$ and use upward compatibility.
    In the base case $w = \varepsilon$, there is nothing to prove.

    Let us now consider a word $w.a$ and $s \trans{w} t \trans{a} r$.
    Let $s \in S$ with $s \preceq s'$.
    By induction, there is a $t'$ such that $s' \trans{w} t'$ and $t \preceq t'$.
    By the upward compatibility of \UWSTSes, we get that there is some $r'$ with $t' \trans{a} r'$ with $r \preceq r'$ as required.
\end{proof}

\restateLemmaHigherBetterDownwardClosure*

\begin{proof}
    We claim that for each $w \in \Sigma^*$, $s \preceq s'$ and $s \trans{w} r$ in $\downclosure{\W}$ implies
    $s' \trans{w} r'$ in $\W$ for some $r'$ with $r \preceq r'$.
    	
	We proceed by induction on $w$.
    In the base case, the statement follows from the fact that the initial configurations of $\downclosure{\W}$ are the configurations in $\downclosure{I}$.
    
    Now consider a word $w.a.$ such that $s \trans{w} t \trans{a} r$ in $\downclosure{\W}$.
    By induction, there is $t'$ such that $s' \trans{w} t'$ in $\W$ and $t \preceq t'$.
    From the definition of transition in $\downclosure{\W}$, we have that $t \trans{a} r''$ in $\W$ for some $r''$ with
    $r \preceq r'' $.
    Since $t \preceq t'$, we can apply upward compatibility to obtain $t' \trans{a} r'$ in $\W$ with $r'' \preceq r'$.
    We conclude $s' \trans{w} t' \trans{a} r'$ in $\W$ with $r \preceq r'' \preceq r'$ as desired.
\end{proof}

\section{Details for Section \ref{sec:expressibility}}
\label{appendix:expressibility}

\restateTheoremSummaryExpr*

\noindent
We formulate and prove a series of lemmas which jointly prove Theorem~\ref{thm:summary-expr}.

\begin{restatable}{rlemma}{restateLemmaUWSTSomegatwodet}
    \label{lem:UWSTSomega2det}
    Every $\omega^2$-\UWSTS is equivalent to a deterministic \UWSTS.
\end{restatable}


\begin{proof}
    Let $\W = (S, T, \preceq, I, F)$ be an arbitrary $\omega^2$-\UWSTS. Wlog.~we assume that $\W = \downclosure{\W}$, i.e.~$\succe{\W}{X}{a}$ is downward closed for every $X$ and $a$, and $\reach{\W}{w}$ is downward closed for every $w$.
    
    We define a deterministic \UWSTS
    $\barW = (\barS, \barT, \subseteq, \barI, \barF)$ essentially as a powerset construction on $\W$.
    Let the configurations $\bar S = \pdown{S}$ be the downward closed subsets of $S$, ordered by inclusion $\subseteq$;
    this is a \wqo by Lemma~\ref{lem:omega2}.
    Let the only initial configuration of $\barW$ be the set $I$ itself, i.e.~let $\barI = \set{I}$.
    (Recall that $I$ is downward closed in~$\downclosure{\W}$.)
    Let the accepting configurations $\barF = \setof{X \in \barS}{X \cap F \neq \emptyset}$ be those downward closed
    subsets $X\in\barS$  which contain at least one accepting configuration from $F$.
    The transition relation $\barT$ is defined by the direct image,
    \begin{align}
    \label{eq:succdet}
    X \trans{a} \succe{\W}{X}{a}
    \end{align}
    which is well-defined as $\succe{\W}{X}{a} \subseteq S$ is always downward closed.
    The equality of the languages of $\W$ and $\barW$ follows directly from the following claim:
    \begin{claiminproof}
        For every $w\in \Sigma^*$, $\reach{\barW}{w} = \set{\reach{\W}{w}}$
        \ .
    \end{claiminproof}
    The claim is shown by induction on the length of $w$.
    The base case is $\barI = \{I\}$, and the step follows by~\eqref{eq:reachsucc} combined with
    the equality $\succe{\barW}{X}{a} = \{\succe{\W}{X}{a}\}$, a reformulation of~\eqref{eq:succdet}.
\end{proof}

\begin{restatable}{rlemma}{restateLemmaDWSTSomegatwodet}
    \label{lem:DWSTSomega2det}
    Every $\omega^2$-\DWSTS is equivalent to a deterministic \DWSTS.
\end{restatable}
\noindent
This can be proven similarly to~Lemma~\ref{lem:UWSTSomega2det}.


\begin{proof}
    Let $\W = (S, T, \preceq, I, F)$ be an $\omega^2$-\DWSTS.
    Wlog.~we assume that $\W = \upclosure{\W}$, i.e., the initial configurations and the transition relation in $\W$ are upward closed.
    A deterministic \DWSTS $\barW = (\barS, \barT, \supseteq, \barI, \barF)$ equivalent to $\W$ can be defined,
    essentially by a powerset construction on $\W$, similar to the proof of Lemma~\ref{lem:UWSTSomega2det}.
    Let the configurations $\barS = \pup{S}$ be the upward closed subsets of $S$, ordered by the superset relation $\supseteq$.
    
    Using the fact that $(\pup{S},\supseteq)$ and $(\pdown{S},\subseteq)$ are isomorphic and Lemma~\ref{lem:omega2}, we obtain that $(\pup{S}, \supseteq)$ is a \wqo. 
    Let the initial configuration of $\barW$ be the set $I$, i.e.~$\barI - \set{ I}$.
    The accepting configurations are those sets containing an accepting configuration in $\W$,
    $\barF = \setof{ X \in \barS }{ X \cap F \neq \emptyset }$.
    
    The transition relation $\barT$ is defined by the direct image:
    \[
    X \trans{a} \succe{\W}{X}{a}
    \]
    which is well-defined as $\succe{\W}{X}{a}$ is upward closed for $X \in \pup{S}$.
    The equality of the languages of $\W$ and $\barW$ follows directly from the following claim:
    \begin{claiminproof}
        For every $w\in \Sigma^*$, $\reach{\barW}{w} = \set{\reach{W}{w}}$.  
    \end{claiminproof}
    The claim is shown by induction on the length of $w$.
    The induction base is $\barI = \{I\}$, and the inductive step follows by~\eqref{eq:reachsucc} combined with
    the equality $\succe{\barW}{X}{a} = \set{\succe{\W}{X}{a}}$.
\end{proof}

Also similarly, but using the finitary downward closed subsets, we prove the following result.

\begin{restatable}{rlemma}{restateLemmaUWSTSfintwodet}
    \label{lem:UWSTSfin2det}
    Every finitely-branching \UWSTS is equivalent to a deterministic \UWSTS.
\end{restatable}


\begin{proof}
    Given a finitely-branching \UWSTS \mbox{$\W = (S, T, \preceq, I, F)$} we define a deterministic one
    \mbox{$\barW = (\barS, \barT, \subseteq, \barI, \barF)$} by a powerset construction.
    We proceed similarly as in the proof of Lemma~\ref{lem:UWSTSomega2det}, but using the finitely-represented downward closed subsets of $S$ instead of all such sets.

    Let configurations be downward closures of finite subsets of $S$, $\barS = \pfindown{S}$.
    They are ordered by inclusion $\subseteq$, which is a \wqo by Lemma~\ref{c:pfin}. 
    Define the unique initial configuration of $\barW$ as $\barI = \set{ \downclosure{I} }$, and note that $I$ is finite.
    The final configurations of $\barW$ are the sets containing a final configuration from $\W$,
    \mbox{$\barF = \setof{X\in\pfindown{S}}{X\cap F \neq \emptyset}$.}
    Let the deterministic transition relation be defined by direct image,
    \[X \trans{a} \downclosure{\succe{\W}{X}{a}}.\]
    Note that $X\in\pfindown{S}$ implies
    $\succe{\W}{X}{a} \in \pfindown{S}$, as $\W$ is assumed to be finitely branching and satisfies upward compatibility.
    The equality of the language of $\W$ and the language of $\barW$ follows by the following claim and the fact that $\W$ and $\downclosure{\W}$ are language-equivalent:
    \begin{claiminproof}
        For every $w\in \Sigma^*$, $\reach{\barW}{w} = \set{\reach{\downclosure{\W}}{w}}$
        \ .
    \end{claiminproof}
    We proceed by induction on $w$.
    
    \noindent
    In the base case, we have $\reach{\barW}{\varepsilon} = \set{ \downclosure{I} }= \set{ \reach{\downclosure{\W}}{w} }$ as desired.
    
    Let us now consider some word $w.a$.
    Using induction and the definition of the transition relation in $\barW$, we have
    \begin{align*}
        \reach{\barW}{w.a}
        &=
        \succe{\barW}{\reach{\barW}{w}}{a}
        \\
        &=
        \succe{\barW}{\set{\reach{\downclosure{W}}{w}}}{a}
        \\
        &=
        \set{ \downclosure{ \succe{\W}{\reach{\downclosure{W}}{w}}{a} } }
    \end{align*}
    
    We claim that we indeed have 
    \[
        \downclosure{ \succe{\W}{\reach{\downclosure{W}}{w}}{a} }
        =
        \reach{\downclosure{\W}}{w.a}
        \ .
    \]
    Let $s \in S$ be such that $s \preceq s'$ for some $s' \in \succe{\W}{\reach{\downclosure{W}}{w}}{a}$.
    This means $s' \trans{a} t'$ in $\W$ for some $t' \in \reach{\downclosure{W}}{w}$.
    By the definition of the transition relation of $\downclosure{\W}$, we conclude $s \trans{a} t$ in $\downclosure{\W}$ which implies $s \in \reach{\downclosure{\W}}{w.a}$.
    
    Let $s \in \reach{\downclosure{\W}}{w.a}$, then we have $s \trans{a} t$ in $\downclosure{\W}$ for some $t \in \reach{\downclosure{\W}}{w}$.
    By the definition of the transition relation of $\downclosure{\W}$, we have that there is some $s'$ with $s \preceq s'$ and $s' \trans{a} t$ in $\W$.
    We obtain $s' \in \succe{\W}{\reach{\downclosure{W}}{w}}{a}$ and conclude \mbox{$s \in \downclosure{ \succe{\W}{\reach{\downclosure{W}}{w}}{a} }$.}
         
    This proves the claim, showing that the language of $\barW$ is equal to the language of $\downclosure{\W}$, which in turn is equal to the language of $\W$.
\end{proof}

\begin{restatable}{rlemma}{restateLemmaDWSTSomegatwofin}
    \label{lem:DWSTSomega2fin}
    Every \DWSTS is equivalent to a finitely branching \DWSTS.
\end{restatable}


\begin{proof}
    Given a \DWSTS $\W = (S, T, \preceq, I, F)$ we define a finitely-branching one
    $\barW = (\barS, \barT, \barpreceq, \barI, \barF)$.
    Wlog.~we assume, $\W = \upclosure{\W}$.
    
    The configurations of $\barW$ are the same as those of $\W$, $\barS = S$.
    The transition relation $\barT$ of $\barW$ is a subset of $T$, where only minimal successors wrt.~$\preceq$ are allowed:
    \begin{align} 
    \label{eq:succmin}
    \succe{\barW}{x}{a} \ \eqdef \ \min (\succe{\W}{x}{a}).
    \end{align}
    The initial configurations of $\barW$ are minimal initial configurations of $\W$, $\barI = \min (I)$.
    As $(S, \preceq)$ is a \wqo, all upward closed sets have a minimal basis.
    Since in $\upclosure{\W}$, the set of successors and the set of initial states are upward closed, $\barT$ is finitely branching and $\barI$ is finite.
    Finally, we put $\barF = F$.
    The equality of the languages of $\W$ and $\barW$ is due to the following equality for all $w\in\Sigma^*$:
    \begin{align}
    \label{eq:reachmin}
    \reach{\barW}{w} \ = \ \min (\reach{\W}{w}).
    \end{align}
    Indeed, as $F$ is downward closed, we have that $\reach{\barW}{w}$ contains a configuration from $F$ if and only if
    $\reach{\W}{w}$ does. 
    
    Finally, the equality~\eqref{eq:reachmin} itself is shown by induction on the length of $w$.
    The base case is $\barI = \min(I)$.
    The induction step follows by~\eqref{eq:reachsucc} and~\eqref{eq:succmin}.
    We start with
    \begin{align*}
    X \ \eqdef \ \reach{\barW}{wa} \
    &\stackrel{\eqref{eq:reachsucc}}{=} \ \succe{\barW}{\reach{\barW}{w}}{a}
    &\stackrel{\eqref{eq:succmin}}{=} \ \min (\succe{\W}{\reach{\barW}{w}}{a})
    \end{align*}
    and use the induction hypothesis to derive
    \mbox{$X = \min (\succe{\W}{\min (\reach{\W}{w})}{a})$}.
    Finally, using downward-compatibility and the assumption that the transition relation in $\W$ is upward closed we observe that 
    for every upward closed set $Y\subseteq S$,
    \[
    \succe{\W}{\min (Y)}{a}
    \ = \ 
    \succe{\W}{Y}{a},
    \]
    which allows us to complete the induction step:

    \mbox{\(
        X \ = \ 
        \min (\succe{\W}{\reach{\W}{w}}{a}) \ \stackrel{\eqref{eq:reachsucc}}{=} \  \min (\reach{\W}{wa})
        \ .
        \)}
\end{proof}

By using similar powerset constructions, we can show the following two results.

\begin{restatable}{rlemma}{restateLemmaUpdown}
    \label{lem:updown}
    Every $\omega^2$-\UWSTS is reverse-equivalent to a deterministic \DWSTS.
\end{restatable}


\begin{proof}
    Given an $\omega^2$-\UWSTS  $\W = (S, T, \preceq, I, F)$,
    we define a deterministic \DWSTS $\barW = (\barS, \barT, \barpreceq, \barI, \barF)$ as follows:

    Configurations are upward closed subsets of $S$, i.e.~$\barS = \pup{S}$, ordered by the superset relation, the reverse of inclusion:
    $U \barpreceq V$ iff $U \supseteq V$. 
    This order is isomorphic to the order $(\pdown{S}, \subseteq)$ of downward closed subset ordered by inclusion, hence a \wqo by Lemma~\ref{lem:omega2}. 
    There is one initial configuration in $\barW$, namely $\barI = \{F\}$. 
    An upward closed subset $U \in \pup{S}$ is final, i.e.~$U \in \barF$, if $U \cap I \neq \emptyset$, i.e.~if $U$ contains some initial configuration from $I$.
    The set $\barF$ is downward closed as required; indeed, if $U \in \barF$ and $U \subseteq V$ then necessarily $V\in \barF$ too.
    The deterministic transition relation is defined using the pre-image:
    \[
    U \trans{a} \prede{\W}{U}{a}
    \ .
    \]
    (Note that $\prede{\W}{U}{a}$ is upward closed whenever $U$ is.)
    Finally, we verify that $\barpreceq$ satisfies the downward compatibility condition: 
    If $U \subseteq V$, then $\prede{\W}{U}{a} \subseteq \prede{\W}{V}{a}$.
    
    From the following claim, we deduce that the language of $\W$ is the reverse of the language of $\barW$.
    \begin{claiminproof}
        For every $w\in \Sigma^*$, $\reach{\barW}{w} = \set{\revreach{\W}{\rev(w)}}$.  
    \end{claiminproof}
    We prove the claim by induction on $w$.
    In the base case, we have
    \begin{align*}
    \reach{\barW}{\varepsilon}
    = \barI = \set{ F }
    = \set{ \revreach{\W}{w} }
    \ .
    \end{align*}
    Now consider a word $w.a$ and note that $\rev(w.a) = a.\rev(w)$.
    Using induction and the definition of $\barT$, we have
    \begin{align*}
    \reach{\barW}{w.a}
    &\stackrel{\eqref{eq:reachsucc}}{=}
    \succe{\barW}{a}{\reach{\barW}{w}}
    \\
    &=
    \succe{\barW}{a}{\set{\revreach{\W}{\rev(w)}}}
    \\
    &=
    \prede{\W}{a}{\set{\revreach{\W}{\rev(w)}}}
    \\
    &\stackrel{\eqref{eq:reachpred}}{=}
    \set{\revreach{\W}{a.\rev(w)}}
    \ .
    \end{align*}
    Thus,
    $w\in \lang{\barW}$
    iff (by the above claim) $\revreach{\W}{\rev(w)} \in \barF$
    iff $\revreach{\W}{\rev(w)} \cap I \neq \emptyset$
    iff $\reach{\W}{\rev(w)} \cap F \neq \emptyset$
    iff \mbox{$\rev(w) \in \lang{\W}$.}
    %
    %
\end{proof}

\begin{restatable}{rlemma}{restateLemmaDownup}
    \label{lem:downup}
    Every $\omega^2$-\DWSTS is reverse-equivalent to a deterministic \UWSTS.
\end{restatable}


\begin{proof}
    Let    $\W = (S, T, \preceq, I, F)$ be the given $\omega^2$-\DWSTS,
    we show how to construct the required deterministic \UWSTS $\barW = (\barS, \barT, \subseteq, \barI, \barF)$ as follows:
    Configurations are downward closed subsets of $S$, i.e.~$\barS = \pdown{S}$, ordered by inclusion $\subseteq$.
    By Lemma~\ref{lem:omega2}, $(\pdown{S},\subseteq)$ is a \wqo.
    
    The initial configuration is given by $\barI = \set{ F } $.
    Note that by definition $F$ is a downward closed set.
    A downward closed set $D \in \barS$ is final if it contains an initial configuration, i.e.
    \[
        \barF = \setof{ D \in \pdown{S} }{ D \in \barS \wedge D \cap I \neq \emptyset}
    \]
    Note that for any $V$, $U \subseteq V$ for some $U \in \barF$ implies $V \cap I \neq \emptyset$, hence $V \in \barF$.
    Hence $\barF$ is upward closed as required.
    The deterministic transition relation is given by
    \[
        U \trans{a} \prede{\W}{U}{a}
    \]
    If $U$ is downward closed then $\prede{\W}{U}{a}$ is also downward closed by downward-compatibility.
    
    Finally, we show that $\barpreceq$ satisfies upward compatibility.
    For this, let $U,V$ be such that  $U \subseteq V$.
    We need to prove that if $U \trans{a} U'$, then there is a $V'$ with $V \trans{a} V'$ and $  U' \subseteq V'$.
    This follows from $U' = \prede{\W}{U}{a} \subseteq  \prede{\W}{V}{a} = V'$.
    
    The equality of the language of $\W$ and the language of $\barW$ follows by the following claim:
    \begin{claiminproof}
    For every $w\in \Sigma^*$, $\reach{\barW}{w} = \set{\revreach{\W}{\rev(w)}}$.  
    \end{claiminproof}
    The proof of the claim is similar to the one in the proof of Lemma~\ref{lem:updown}.
    Thus,
    $w\in \lang{\barW}$
    iff (by the above claim) \mbox{$\revreach{\W}{\rev(w)} \in \barF$}
    iff \mbox{$\revreach{\W}{\rev(w)} \cap I \neq \emptyset$}
    iff \mbox{$\reach{\W}{\rev(w)} \cap F \neq \emptyset$}
    iff \mbox{$\rev(w) \in \lang{\W}$.}
\end{proof}

\section{Details for Section~\ref{sec:separability}}

\restateLemmaSim*

\begin{proof}
        
        \mbox{}
    
        \noindent(1)
            We show that for every $s \in \reach{\W}{w}$, there is some \mbox{$(s, s') \in \reach{\prodW}{w}$}.
            
            Indeed, 
            $s \in \reach{\W}{w}$ implies \mbox{$(s, s') \in \reach{\prodW}{w}$} for \mbox{$s' = \reach{\W'}{w}$.}
            
       \noindent(2)
            We show that for every $(s, s') \in \reach{\prodW}{w}$, there is some $(r, r') \in \reach{\A}{w}$ with $(s, s') \prodpreceq (r, r')$.
           
            We proceed by induction on $w$.
            In the base case, we have  that any $(s,s') \in \prodI = \reach{\prodW}{\varepsilon}$ is dominated by some $(r,r') \in Q$ since $I \subseteq \downclosure{Q}$, Property~\ref{eq:iiI}.
            By the definition of $Q_I$, we have \mbox{$(r,r') \in Q_I = \reach{\A}{\varepsilon}$}.
            
            Now consider $(s,s') \in \reach{\prodW}{w.a}$.
            By definition, there is $(\tilde{s},\tilde{s}') \in \reach{\prodW}{w}$ and $(\tilde{s},\tilde{s'}) \trans{a} (s,s')$ in $\prodW$.
            Applying induction, we obtain $(\tilde{r},\tilde{r}') \in \reach{\A}{w}$ with \mbox{$(\tilde{s},\tilde{s}') \prodpreceq (\tilde{r},\tilde{r}')$.}
            Using the upward-compatibility of $\prodW$, there is a transition \mbox{$(\tilde{r},\tilde{r}') \trans{a} (t,t')$} in $\prodW$ such that $(s,s') \prodpreceq (t,t')$.
            Since $(r,r') \in Q$, and $\downclosure{Q}$ is closed under taking successors in $\prodW$, Property~\ref{eq:iiSucc}, we have that there is $(r,r') \in Q$ such that $(t,t') \prodpreceq (r,r')$.
            We have $(s,s') \prodpreceq (r,r')$ be transitivity.
            To complete the proof, it remains to argue that there is a transition $(\tilde{r},\tilde{r}') \trans{a} (r,r')$ in $\A$.
            To this end, we instantiate the definition of the transition relation, using $(t,t') \prodpreceq (r,r')$.
\end{proof}

\restateLemmaInvsimWprime*

\begin{proof}
    We proceed by induction on $w$.
    The base case follows by the definition of the initial states in $\A$.
    For the induction step, consider $(s, s') \in \reach{\A}{w.a}$, which means 
    there is $(r,r') \in \reach{\A}{w}$ with $(r,r') \trans{a} (s,s')$ in $\A$.
    By definition of the transition relation in the automaton, there is $(t,t')$ such that
    $(r,r') \trans{a} (t,t')$ in $\prodW$ and $ (t,t') \prodpreceq (s,s')$.
    By definition of the transition relation in $\prodW$, we have $r' \trans{a} t'$ in $\W'$.
    
    We apply induction to $(r,r')$ and get \mbox{$\reach{\W'}{w} \preceq' r'$.}
    In $\W'$, we have \mbox{$\reach{\W'}{w} \trans{a} \reach{\W'}{w.a}$.}
    Using upward compatibility of $\W'$, $r'$ can simulate this transition.
    Since $\W'$ is deterministic, it is in fact simulated by $r' \trans{a} t'$ and we conclude $\reach{\W'}{w.a} \preceq' t'$.
    Hence, \mbox{$\reach{\W'}{w.a}  \preceq' t' \preceq' s'$}.
\end{proof}


\begin{proof}
    Our goal is to apply Theorem~\ref{thm:core}.
    Consider an arbitrary \DWSTS $\W = (S, T, \preceq, I, F)$ and a deterministic one $\W' = (S', T', \preceq', I', F')$.
    We start with the observation that the inversed versions of $\W$ and $\W'$, namely with the orders $\preceq^{-1}$ and 
    $(\preceq')^{-1}$ and denoted by $\W^{-1}$ and $(\W')^{-1}$, are \ULTSes.
    We claim that these \ULTSes satisfy the assumptions of Theorem~\ref{thm:core}.
    The language of \mbox{$\prodW^{-1}=\W^{-1} \times (\W')^{-1}$} is empty since the language of $\prodW=\W \times \W'$ is empty and inversion does not change the language, $\lang{\W} = \lang{\W^{-1}}$ and similar for $\W'$. 
    Inversion also does not influence determinism. 
    
    It remains to find an inductive invariant of $\prodW^{-1}$ that is finitely represented.
    We claim that 
    \begin{align*}
    X\ =\ \downclosureinv{\reachall{\prodW^{-1}}}
    \end{align*}
    is a suitable choice.
    The subscript indicates that the downward closure is computed relative to the quasi order of $\prodW^{-1}$. 
    The proof of Lemma~\ref{Lemma:Invariants} shows that $X$ is an inductive invariant.
    For the finite representation, note that inversion does not change the transition relation. 
    Hence, $\prodW$ and $\prodW^{-1}$ reach the same configurations, 
    \begin{align*}
    \reachall{\prodW^{-1}}=\reachall{\prodW}=Z\ .
    \end{align*}
    With the definition of inversion, $X=\downclosureinv{Z} = \upclosure{Z}$ holds. 
    Moreover, $\upclosure{Z} = \upclosure{\min(Z)}$, with minimum and upward closure computed relative to $\prodW$. 
    Since the configurations of $\prodW$ are well quasi ordered, $\min(Z)$ is finite. 
    Another application of inversion yields $X=\upclosure{\min(Z)}=\downclosureinv{\min(Z)}$.  
    Hence, $X$ is a finitely-represented downward-closed subset of $\prodW^{-1}$. 
    
    By Theorem~\ref{thm:core}, the languages of $\W^{-1}$ and $(\W')^{-1}$ are regular separable and so are the languages of $\W$ and $\W'$. 
\end{proof}

\restateLemmaIdcompl*

\begin{proof}
    The ideal completion is a \ULTS.
    Indeed, the set of final configurations $\barF$ is upward-closed since the ideals are ordered by inclusion. 
    The transition relation still satisfies upward compatibility by Lemma~\ref{lem:idealsirred}.
    
    The language is preserved due to the following invariant:

    \begin{claiminproof}
    For every $w\in \Sigma^*$, \(\bigcup \reach{\idcompl{\W}}{w} \ = \ \downclosure{\reach{\W}{w}}\)\ .
    \end{claiminproof}
    We prove the claim by induction on $w$.
    
    \noindent
    In the base case, we have
    \begin{align*}
    \bigcup \reach{\idcompl{\W}}{\varepsilon}
    &\, = \, \bigcup \idcompl{I}
    \\
    &\, = \,  \bigcup \idec{S}{\downclosure{I}}
    \\
    &\stackrel{\eqref{eq:union}}{=} \downclosure{I}
    \\
    &\, = \,  \downclosure{ \reach{\W}{\varepsilon}}
    \end{align*}
    as required.
    
    Consider $w.a$. 
    Using upward compatibility, induction, and the fact that unions commute with taking the successor in $\W$, we obtain
    \begin{align*}
    \downclosure{\reach{\W}{w.a}}
    &= \downclosure{\succe{\W}{\reach{\W}{w}}{a}}
    \\
    &= \downclosure{
        \succe{\W}{\downclosure{\reach{\W}{w}}}{a}
    }
    \\
    &= \downclosure{
        \succe{\W}{\bigcup \reach{\idcompl{\W}}{w}}{a}
    }
    \\
    &= \downclosure{
        \succe{\W}{\bigcup_{Y \in \reach{\idcompl{\W}}{w}}Y}{a}
    }
    \\
    &=
    \bigcup_{Y \in \reach{\idcompl{\W}}{w}}
    \downclosure{
        \succe{\W}{Y}{a}
    }
    \ .
    \end{align*}
    Using Equation~\eqref{eq:union}, the last expression is equal to
    \[
    \bigcup_{Y \in \reach{\idcompl{\W}}{w}}
    \bigcup
    \idec{S}{
        \downclosure{
            \succe{\W}{Y}{a}
        }
    }\ .
    \]
    By the above equality and the definition of the transition relation in $\idcompl{\W}$, we obtain
    \begin{align*}
    &\bigcup_{Y \in \reach{\idcompl{\W}}{w}}
    \bigcup
    \downclosure{
        \succe{\W}{Y}{a}
    }
    \\
    =\ &
    \bigcup
    \bigcup_{Y \in \reach{\idcompl{\W}}{w}}
    \succe{\idcompl{\W}}{Y}{a}
    \\
    =\ &
    \bigcup
    \succe{\idcompl{\W}}{\reach{\idcompl{\W}}{w}}{a}
    \\
    =\ &
    \bigcup
    \reach{\idcompl{\W}}{w.a}
    \end{align*}
    as desired.
    
    Suppose now that $\W= (S, T, \preceq, I, F)$ is deterministic.
    In particular, $I = \set{i}$ for some $i \in S$.
    Then $\downclosure{\set{i}}$ is the unique initial configuration of $\idcompl{\W}$. 
    To prove that $\succe{\idcompl{\W}}{X}{a}$ contains a unique element, we show that $\downclosure{\succe{\W}{X}{a}}$ is already an ideal. 
    To this end, it is sufficient to show that whenever $X$ is directed, $\succe{\W}{X}{a}$ is directed. 
    Consider $r,r' \in \succe{\W}{X}{a}$.
    Then there are $s,s' \in X$ with $s \trans{a} r$ and $s' \trans{a} r'$.    
    As $X$ is directed, there is an element $\tilde{s}$ with $s \preceq \tilde{s}$ and $s' \preceq \tilde{s}$.
    Using upward compatibility, $\tilde{s}$ can simulate the transitions of $s$ and $s'$.
    Since $\W$ is deterministic, there is in fact a unique $\tilde{r} \in \succe{\W}{X}{a}$ with $\tilde{s} \trans{a} \tilde{r}$ and $r \preceq \tilde{r}$ as well as $r' \preceq \tilde{r}$, as desired.
\end{proof}

\restateLemmaIsoprod*

\begin{proof}
    With Lemma~\ref{lem:prodideals}, 
    as an isomorphism between $\idcompl{\W} \times \idcompl{\W'}$ and $\idcompl{\W \times \W'}$
    take the function $I, J \mapsto I\times J$ that maps a pair of ideals to their  product. 
    This is an isomorphism as the transition relation in the ideal completion of a \UWSTS is defined by direct image, 
    and direct image commutes with product.
\end{proof}

\section{Details for Section~\ref{sec:bound}}

\subsection{Proof of Theorem~\ref{Theorem:UpperPN}}

\restatePropositionUpperbounddet*

\begin{proof}
    Let $N_1$ and $N_2$ be the given Petri nets with a total of $d\in \Nat$ places. 
    Let $\W=\wstsof{N_1\times N_2}$. 
    Since $\lang{\W}$ is empty, $X = \Nat^d \setminus \revreachall{\W}$ is an inductive invariant of $\W$ by the proof of Lemma~\ref{Lemma:Invariants}. 
    By Proposition~\ref{prop:inducedinvariant}, $\downclosure{Y}$ with $Y = \idec{\Nat^d}{X}$ is a finitely-represented inductive invariant in the ideal completion $\idcompl{\W}$.  
    By Lemma~\ref{Lemma:ProductNet} and Lemma~\ref{lem:isoprod}, we have
    \begin{align*}
    \idcompl{\wstsof{N_1\times N_2}}\ \iso\ \idcompl{\wstsof{N_1}\times \wstsof{N_2}}\ \iso\ \idcompl{\wstsof{N_1}}\times \idcompl{\wstsof{N_2}}\ .
    \end{align*}
    Hence, as $N_2$ is deterministic, we can construct a separating finite automaton with states $Y$ by Theorem~\ref{thm:core}.
    
    It remains to prove that the cardinality of $Y$ is at most doubly-exponential.
    To this end, we invoke Theorem~\ref{thm:BozelliGanty} and consider a representation $\revreachall{\W} = \upclosure{\set{v_1, \ldots, v_k}}$.
    We have
    \begin{align*}
    \Cmplment
    &= \Nat^d \setminus \upclosure{\set{ v_1, \ldots, v_k }}\\
    &= \Nat^d \setminus \bigcup_{i \in \interval{k}} \upclosure{v_i}\\
    &= \bigcap_{i \in \interval{k}} \left(\Nat^d \setminus \upclosure{v_i} \right)\\
    &= \bigcap_{i \in \interval{k}}\bigcup_{j_i \in \interval{d}} \downclosure{u_{<v_i(j_i)}}
    \  .
    \end{align*}  
    where the ideal representatives $u_{<v_i(j_i)}$ are constructed as in Lemma~\ref{lem:natideals}(3).     
    We use distributivity to rewrite this expression as 
    \begin{align*}
    \bigcap_{i \in \interval{k}} \bigcup_{j_i \in \interval{d}} \downclosure{u_{<v_i(j_i)}}
    &
    =
    \bigcup_{\vec{j} \in \interval{d}^k}
    \bigcap_{i\in \interval{k}}
    \downclosure{u_{<v_i(\vec{j}(i))}}
    \ .
    \end{align*}
    By Lemma~\ref{lem:natideals}(1) and~(2), applied inductively, we obtain that each intersection 
    \begin{align*}
    \bigcap_{i=1}^{k} \downclosure{u_{<v_i(\vec{j}(i))}}\quad\text{is an ideal}\quad \downclosure{u_{\vec{j}}}\ .
    \end{align*}
    This means the ideal decomposition of $\Cmplment$ consists of at most $d^k \leq d^g$ many ideals, with $g$ as defined in Theorem~\ref{thm:BozelliGanty}.
    
    This bound is triply-exponential.
    We improve by observing that $\norm{u_{\vec{j}}} \leq \norm{\set{v_1, \ldots, v_k}} \leq g$ for all $\vec{j}$.
    (Here, the infinity norm is extended to $\Nat_\omega^d$ by treating $\omega$-components as zero.)
    Consequently, all non-$\omega$ components are bounded by $g$, and we only have $(g+2)^d$ many such vectors in $\Nat_\omega^d$. 
    
    The ideal decomposition of $\Cmplment$ thus consists of at most
    \begin{align*}
    h=\left(
    \left(
    \card{T}
    \left(
    \norm{F}
    +
    \norm{M_0}
    +
    \norm{M_f}
    +
    2
    \right)
    \right)^{2^{\bigO{d \cdot \log d}}}
    + 2
    \right)^d
    \end{align*}
    many ideals. 
    Note that even if $ \norm{F} + \norm{M_0} + \norm{M_f}$ are exponential in $\card{N_1}+\card{N_2}$ (due to the binary encoding of values), $h$ is still doubly exponential in the size of the given Petri nets.
\end{proof}

\restateLemmaPNDeterminization*

\begin{proof}
    We show that for the Petri nets $\NAlambda$ and $\NBfree$ that we have constructed,
    \[
        \lang{\NA \times \NB} = \lambda ( \lang{ \NAlambda \times \NBfree })
    \]
    holds.
        
    Let $w$ be in the left-hand side.
    Consider the corresponding computations $u_1$ and $u_2$ of $\NA$ and $\NB$, respectively.
    $u_2$ can be seen as a computation of $\NBfree$ with $\lambda(u_2) = w$ as desired.
    Consider the computation $u_1'$ of $\NAlambda$ that is obtained as follows:
    If at some position $i$, $u_1$ uses transition $t_1 \in \TA$ and $u_2$ uses transition $t \in \TB$ (where $\lambdaA(t_1) = \lambda (t)$ has to hold), we let $u_1'$ use transition $t_1^{t}$.
    Note that in $\NAlambda$, this transition has label $t$.
    Indeed, $u_1'$ synchronizes with $u_2$ as desired, proving $u_2 \in \lang{ \NAlambda \times \NBfree }$ and thus $\lambda(u_2) = w$ is in the right-hand side.
    
    Let $\lambda(u_2)$ be in the right hand side and let $u_1'$, $u_2$ be the corresponding computations.
    $u_2$ is already a computation of $\NB$ with labeling $\lambda(u_2)$ as desired.
    We define the computation $u_1$ of $\NA$ as follows:
    If at some positions $i$, $u_1'$ uses some transition $t_1^{t}$, we define $u_1$ to use transition $t$ at this position.
    Note that $\lambdaA(t_1) = \lambda(t)$ has to hold.
    The computations synchronize as desired, proving that $\lambda(u_1') = \lambda(u_2)$ is in the left-hand side.
\end{proof}

\restateLemmaBisSep*

\begin{proof}
    We have $\lang{\NBfree}\subseteq \lang{\complA}$ and hence
   \begin{align*}
 \lang{\NB} = \lambda(\lang{\NBfree}) \subseteq \lambda(\lang{\complA})=\ 
 \lang{\lambda(\complA)}=\lang{\B}.
    \end{align*}

    As for disjointness, assume $w\in \lang{\NA}\cap\lang{\B}$.
    Then there is an accepting computation $u$ in $\NA$ with $\lambdaA(u)=w$ and $v\in\lang{\complA}$ with $\lambda(v)=w$.
    We inductively construct a $v$-labeled accepting computation in $\NAlambda$.
    This will contradict $\lang{\NAlambda} \cap \lang{\complA} = \emptyset$.
    Whenever $u$ uses some transition $t_1$ and $v$ uses $t\in \TB$, use the transition $t_{1}^{t}$ of $\NAlambda$.
    Since we have $\lambdaA (t_1) = \lambda (t)$, the transition $t_{1}^{t}$ indeed exists in $\NAlambda$.
    Because the behavior of $t_1$ and $t_{1}^{t}$ is the same, the resulting computation of $\NAlambda$ is still accepting.
\end{proof}

\subsection{Proof of Theorem~\ref{Theorem:LowerPN} / Proposition~\ref{prop:lowerbound}}

\restatePropositionLowerbound*
\noindent
To construct the required nets, we make use of Lipton's proof of $\EXPSPACE$-hardness for coverability~\cite{Lipton}.
We will not need the precise construction, the following lemma gives a specification of Lipton's Petri nets that is enough for our purposes.

\begin{lemma}[Lipton~\cite{Lipton}]
    \label{Lemma:Lipton}
    For every $n\in\Nat$, 
    \begin{enumerate}[a)]
        \item there is a Petri net $\ninc(n)$ of size polynomial in $n$ with places $\phaltinc$ and $\pout$ such that
        any computation leading to a marking $M$ with $M(\phaltinc) = 1$ has $M(\pout) = 2^{2^n}$.
        \item there is $\ndec(n)$ of size polynomial in $n$ with places $\phaltdec$ and $\pin$ such that there is a computation leading to a marking $M$ with $M(\phaltdec) = 1$ if and only if $M_0(\pin) \geq 2^{2^n}$. 
    \end{enumerate}
    We can assume all transitions in these Petri nets carry label $\llipton$.
\end{lemma}
\noindent
Using the lemma, we create for $x \in \binary$ a new Petri net $N_x(n)$ over the alphabet $\set{ 0, 1, \llipton, \lphase}$ whose language is essentially $\mathcal{L}_{x@2^{2^n}}$, where
\[
    \mathcal{L}_{x@k} = \setof{ w \in \binary^{\geq k}   }{ \text{ the } k\text{-last letter in } w \text{ is } x}
\]

A computation of $N_x(n)$ consists of four successive phases.
Four so-called control places $p_1, \ldots, p_4$  indicate the current phase of the computation, and all transitions of a specific phase check that the corresponding control place carries a token.
\begin{enumerate}
    \item[1.]
    In the first phase, $\ninc(n)$ is used to create $2^{2^n}$ tokens on $\pout$ and one token on $\phaltinc$. The phase ends by firing a $\lphase$-labeled transition that checks for a token on $\phaltinc$ and moves the token from control place $p_1$ to $p_2$.
    Note that all transitions are labeled $\llipton$ up to this point.
    \item[2.]
    In the second phase, there are two transitions labeled by $0$ resp.~$1$ which only check that the control place carries a token.
    
    They create an arbitrary sequence in $\binary^*$, corresponding to the part of the word before the last $2^{2^n}$ letters.
    \item[$x$.]
    The second phase ends with an $x$-labeled transition that moves the token from control place $p_2$ to $p_3$. It consumes one token from $\pout$ and generates one token on $\pin$.
    \item[3.]
    In the third phase, there are two transitions labeled by $0$ resp.~$1$, each moving one token from $\pout$ to $\pin$. 
    The phase ends by firing a $\lphase$-labeled transition that moves the token from $p_3$ to $p_4$.
    \item[4.]
    In the fourth phase, $\ndec(n)$ is used to check that the number of tokens on $\pin$ at the beginning of the phase is at least $2^{2^n}$.
\end{enumerate}
The initial marking for $N_x(n)$ assigns a token to $p_1$ as well as the necessary initial tokens to $\ninc(n)$ and $\ndec(n)$. 
The final marking requires a token on $p_4$ and a token on the place $\phaltdec$ of $\ndec(n)$.

We claim that the language of $N_x(n)$ is
\[
\lang{ N_x(n) }\ =\ \mathcal{L}'.\lphase.\mathcal{L}_{x@2^{2^n}}.\lphase.\mathcal{L}''\ \qquad \text{with }\mathcal{L}',\mathcal{L}'' \subseteq \llipton^*\ .
\]
With the control places $p_1,\ldots, p_4$, words in $\lang{ N_x(n) }$ clearly have the shape $w_{\inc}.c.w.c.w_{\dec}$, where $w_{\inc}$ and $w_{\dec}$ are computations in Lipton's Petri nets.
We argue that the $2^{2^n}$-last letter in $w$ is $x$. 
After running $\ninc(n)$, we have $2^{2^n}$ tokens on $\pout$.
This means the $x$-labeled transition that ends the second phase has to be fired at most $2^{2^n}$ letters before the end of $w$, because any transition fired in the third phase consumes a token from $\pout$. 
Since any transition fired during the third phase also produces a token on $\pin$, and we check for $2^{2^n}$ tokens on this place in phase four, the $x$-labeled transition has to be fired at least $2^{2^n}$ letters before the end of $w$.

Since $\ninc(n)$ and $\ndec(n)$ are of size polynomial in $n$, and in $N_x(n)$ we only add a constant number of transitions and places, and a polynomial number of entries to the flow function, it is clear that $N_x(n)$ is also polynomially-sized.

It remains to argue that $\lang{ N_0(n) }$ and $\lang{ N_1(n) }$ cannot be separated by a DFA of less than triply-exponential size.
Recall that the languages of each $N_x(n)$ is
\[
    \mathcal{L}'.\lphase.\mathcal{L}_{x@2^{2^n}}.\lphase.\mathcal{L}''
    \ .
\]
Since the languages are not distinguishable in their $\mathcal{L}'$ prefix and $\mathcal{L}''$ suffix, the separator has to distinguish $\mathcal{L}_{0@2^{2^n}}$ from $\mathcal{L}_{1@2^{2^n}}$.  
These languages partition $\binary^{\geq 2^{2^n}}$, so the separator has to incorporate a DFA for $\mathcal{L}_{0@2^{2^n}}$.
It is a classic result from automata theory that any DFA for $\L_{x@m}$ needs to have at least $2^m$ states~\cite{Kozen1997}.
Intuitively, a DFA for $\L_{x@m}$ cannot guess the end of the word.
It always needs to store the last $m$ bits of the input that it has processed so far, for which there are $2^m$ possibilities.
For the sake of completeness, we give a formal proof.

\begin{proposition}
\label{prop:lxat22n}
    Any DFA $\A$ such that $\lang{N_0(n)} \subseteq \lang{\A}$ and
    $\lang{\A} \cap \lang{N_1(n)} = \emptyset$ needs to have at least $2^{2^{2^n}}$ many states.
\end{proposition}
\begin{proof}
    For ease of notation, we define $k=2^{2^n}$.
    Assume towards a contradiction that $\A$ has strictly less than $2^k$ states.
    We consider the set $\calB = \binary^{k}$ of all sequences over $\binary$ of length exactly $k$.
    We have that $\card{\calB} = 2^k$.
    
    Let $w_{\inc} \in \llipton^*$ be a word corresponding to computation for $N_{\inc}(n)$ that creates $k$ many tokens on places $\pout$ and one token on $\phaltinc$.
    Similarly, let $w_{\dec} \in \llipton^*$ be a word corresponding to a computation for $N_{\dec}(n)$ that creates a token on $\phaltdec$ (assuming that $\pin$ contains at least $k$ tokens).
    
    For any $w = w_1 \ldots w_{k} \in \calB$, $w_{\inc}.c.w.c.w_{\dec}$ is contained in the language of $N_x(n)$ for exactly one $x \in \binary$, namely for $x = w_1$.
    For each $w \in \calB$, we denote by $q_w$ the unique control state in which the DFA $\A$ is after processing $w_{\inc}.c.w$.
    Since $\A$ has strictly less than $2^k$ states, but $\calB$ has $2^k$ elements, there are distinct $w,w' \in \calB$ such that $q_w = q_{w'}$.
    Since $w \neq w'$, there is some bit $i \geq 1$ such that $w_i \neq w_i'$.
    We assume without loss of generality $w_i = 0, w_i' = 1$.
    Let us define $w_{\mathit{fill}} = 0^{i-1}$.
    Since $w_{\inc}.c.w$ and $w_{\inc}.c.w'$ lead to the same state, also
    \(
        w_{\inc}.c.w.w_{\mathit{fill}}.c.w_{\dec}
        \text{ and }
        w_{\inc}.c.w'.w_{\mathit{fill}}.c.w_{\dec}
    \)
    lead to the same state.
    Thus, either both or none of these words is accepted by $\A$.
    To complete the proof, note that the length \mbox{$i-1$} of $w_{\mathit{fill}}$ was chosen such that $w_i$ is the $k$-last bit of $w.w_{\mathit{fill}}$.
    Therefore, we have $w_{\inc}.c.w.w_{\mathit{fill}}.c.w_{\dec} \in \lang{N_0(n)} \subseteq \lang{\A}$ since the $k$-last bit is $0$, and $w_{\inc}.c.w'.w_{\mathit{fill}}.c.w_{\dec} \in \lang{N_1(n)}$ since the $k$-last bit is $1$.
    We conclude $w_{\inc}.c.w'.w_{\mathit{fill}}.c.w_{\dec} \not\in \lang{\A}$ since \mbox{$\lang{N_1(n)} \cap \lang{\A} = \emptyset$}, a contradiction.
\end{proof}

\end{document}